\documentclass[preprint,12pt,authoryear]{elsarticle}

\usepackage{amsmath}
\usepackage{amssymb}
\usepackage{amsfonts}
\usepackage{amsthm}

\usepackage{mathtools}
\usepackage{enumitem}
\usepackage{comment}
\usepackage{braket}

\usepackage[T1]{fontenc}

\usepackage{dsfont}
\usepackage{mathrsfs}
\usepackage[%
cal=cm,
]
{mathalfa}

\usepackage{accents}

\usepackage[dvipsnames,svgnames]{xcolor}
\colorlet{MyBlue}{DodgerBlue!60!Black}
\colorlet{MyGreen}{DarkGreen!85!Black}

\usepackage{subfigure}
\usepackage{tikz}
\usetikzlibrary{calc,patterns}
\usepackage{changepage}

\usepackage{acronym}
\usepackage{latexsym}
\usepackage{longtable}
\usepackage{wasysym}
\usepackage{xspace}
\usepackage[normalem]{ulem}

\usepackage{hyperref}
\hypersetup{
colorlinks=true,
linktocpage=true,
pdfstartview=FitH,
breaklinks=true,
pdfpagemode=UseNone,
pageanchor=true,
pdfpagemode=UseOutlines,
plainpages=false,
bookmarksnumbered,
bookmarksopen=false,
bookmarksopenlevel=1,
hypertexnames=true,
pdfhighlight=/O,
urlcolor=MyBlue!60!black,linkcolor=MyBlue!70!black,citecolor=DarkGreen!70!black, % <--- for screen
pdftitle={},
pdfauthor={},
pdfsubject={},
pdfkeywords={},
pdfcreator={pdfLaTeX},
pdfproducer={LaTeX with hyperref}
}

\numberwithin{equation}{section}  %\numberwithin{theorem}{section}
\usepackage[sort&compress,capitalize,nameinlink]{cleveref}
\crefname{app}{Appendix}{Appendices}

\crefrangeformat{equation}{\upshape(#3#1#4)\textendash(#5#2#6)}

\usepackage[textwidth=30mm]{todonotes}

\newcommand{\debug}[1]{#1}

\theoremstyle{plain}
\newtheorem{theorem}{Theorem}[section]

\newtheorem*{corollary*}{Corollary}
\newtheorem{lemma}[theorem]{Lemma}

\newtheorem{conjecture}[theorem]{Conjecture}

\theoremstyle{definition}
\newtheorem{definition}[theorem]{Definition}
\newtheorem*{definition*}{Definition}

\newtheorem*{hypothesis*}{Hypothesis}

\theoremstyle{remark}
\newtheorem{remark}[theorem]{Remark}
\newtheorem*{remark*}{Remark}
\newtheorem*{notation*}{Notational remark}
\newtheorem{example}[theorem]{Example}

\usepackage{lineno}

\DeclareMathOperator{\Prob}{\mathsf{\debug{P}}}
\DeclareMathOperator{\Expect}{\mathsf{\debug{E}}}

\newcommand{\cost}{\debug c}
\newcommand{\costprof}{\boldsymbol{\cost}}
\newcommand{\Cost}{\debug C}

\newcommand{\eqcost}{\debug \lambda}
\newcommand{\eqcostprof}{\boldsymbol{\eqcost}}

\newcommand{\eq}[1]{#1^{*}}
\newcommand{\opt}[1]{\widetilde#1}
\newcommand{\activ}[1]{\widehat#1}

\DeclareMathOperator{\SC}{\mathsf{\debug{SC}}}
\DeclareMathOperator{\PoA}{\mathsf{\debug{PoA}}}

\newcommand{\graph}{{\debug G}}
\newcommand{\vertices}{\mathcal{\debug V}}
\newcommand{\edges}{\mathcal{\debug E}}

\newcommand{\vertex}{\debug v}

\newcommand{\edge}{\debug e}

\newcommand{\source}{\mathsf{\debug O}}

\newcommand{\sink}{\mathsf{\debug D}}

\newcommand{\rate}{\debug \mu}

\newcommand{\rateprof}{\boldsymbol{\rate}}

\newcommand{\flow}{\debug f}
\newcommand{\flows}{\mathcal{\debug F}}
\newcommand{\flowprof}{\boldsymbol{\flow}}

\newcommand{\load}{\debug x}

\newcommand{\loads}{\mathcal{\debug X}}
\newcommand{\loadprof}{\boldsymbol{\load}}

\newcommand{\nRoutes}{\debug P}
\newcommand{\routes}{\mathcal{\debug \nRoutes}}
\newcommand{\route}{\debug p}
\newcommand{\routealt}{\debug q}

\newcommand{\argdot}{\,\cdot\,}

\newcommand{\diff}{\ \textup{d}}

\newcommand{\ie}{i.e.,\ }

\newcommand{\zerovec}{\boldsymbol{\debug 0}}

\DeclareMathOperator{\Ker}{Ker}
\DeclareMathOperator{\Range}{Range}

\newcommand{\lagrang}{\mathcal{\debug L}}

\newcommand{\odpair}{\debug h}
\newcommand{\ODpairs}{\debug{\mathcal H}}

\newcommand{\edgepath}{\debug \delta}
\newcommand{\edgepaths}{\debug \Delta}
\newcommand{\sumflow}{\debug S}
\newcommand{\potential}{\debug \Phi}
\newcommand{\eqcostroute}{\debug \cost^{*}}
\newcommand{\neigh}{\debug N}

\newcommand{\pertx}{\debug \xi}
\newcommand{\pertxprof}{\boldsymbol{\pertx}}
\newcommand{\pertf}{\debug \omega}
\newcommand{\pertfprof}{\boldsymbol{\pertf}}
\newcommand{\multod}{\debug m}
\newcommand{\multodprof}{\boldsymbol{\multod}}
\newcommand{\valueV}{\debug V}
\newcommand{\multedge}{\debug \eta}
\newcommand{\multedgeprof}{\boldsymbol{\multedge}}
\newcommand{\multpath}{\debug \nu}
\newcommand{\multpathprof}{\boldsymbol{\multpath}}
\newcommand{\nOD}{\debug H}

\newcommand{\npaths}{\debug P}
\newcommand{\odrate}{\debug r}
\newcommand{\odrateprof}{\boldsymbol{\odrate}}

\newcommand{\regime}{\mathcal{\debug R}}

\newcommand{\yvar}{\debug y}

\newcommand{\thfunc}{\debug \theta}

\newcommand{\inter}{\debug I}

\newcommand{\eqcostedge}{\debug \tau}

\newcommand{\yvec}{\boldsymbol{\debug y}}

\newcommand{\zvar}{\debug z}

\newcommand{\onevec}{\boldsymbol{\debug 1}}
\newcommand{\canbasis}{\boldsymbol{\debug \gamma}}

\newcommand{\Rpos}{\mathbb R_+}
\newcommand{\vecw}{\debug w}
\newcommand{\vecz}{\debug z}

\newcommand{\jac}{\debug J}

\newcommand{\Amatr}{\debug A}
\newcommand{\Bmatr}{\debug B}
\newcommand{\dvar}{\debug d}
\newcommand{\dprof}{\boldsymbol{\dvar}}
\newcommand{\multA}{\debug \alpha}
\newcommand{\multAprof}{\boldsymbol{\multA}}
\newcommand{\multB}{\debug \beta}
\newcommand{\multBprof}{\boldsymbol{\multB}}
\newcommand{\var}{\debug t}

\newcommand{\prim}{\mathsf{\debug P}}

\newcommand{\reals}{\mathbb{R}}

\DeclareMathOperator*{\union}{\bigcup}

\DeclarePairedDelimiter{\braces}{\{}{\}}

\DeclarePairedDelimiter{\parens}{(}{)}

\DeclarePairedDelimiterX{\inner}[2]{\langle}{\rangle}{#1,#2}
\DeclarePairedDelimiterX{\setdef}[2]{\{}{\}}{#1:#2}

\DeclarePairedDelimiterXPP{\probof}[1]{\Prob}{(}{)}{}{%

#1}

\DeclarePairedDelimiterXPP{\exof}[1]{\Expect}{[}{]}{}{%

#1}

\newacro{ACG}{atomic congestion game}
\newacro{ACGSD}{atomic congestion game with stochastic demand}
\newacro{CRG}{constrained routing game}
\newacro{PoA}{price of anarchy}
\newacro{PoS}{price of stability}
\newacro{SC}{social cost}
\newacro{SEC}{social expected cost}
\newacro{SO}{social optimum}
\newacro{SOC}{socially optimum cost}

\newacro{PNMC}{parallel-network with multiple-commodities}
\newacro{MES}{monotonic equilibrium selection}

\newacro{NE}{Nash equilibrium}
\newacroplural{NE}[NE]{Nash equilibria}
\newacro{BNE}{Bayesian Nash equilibrium}
\newacroplural{BNE}[BNE]{Bayesian Nash equilibria}
\newacro{PNE}{pure Nash equilibrium}
\newacroplural{PNE}[PNE]{pure Nash equilibria}

\newacro{WE}{Wardrop equilibrium}
\newacroplural{WE}[WE]{Wardrop equilibria}

\newacro{KKT}{Karush\textendash Kuhn\textendash Tucker}
\newacro{OD}[OD]{origin-destination}
\newacro{BPR}{Bureau of Public Roads}

\newacro{SP}{series-parallel}
\newacro{CSP}{constrained series-parallel}

\journal{arXiv
}

\begin{document}

\begin{frontmatter}



\title{Phase Transitions of the Price-of-Anarchy Function in Multi-Commodity Routing Games}


\author[labelRoberto]{Roberto Cominetti}

\affiliation[labelRoberto]{organization={Facultad de Ingenier{\'\i}a y Ciencias, Universidad Adolfo Ib\'a\~nez},%Department and Organization
            addressline={Diagonal las Torres 2640}, 
            city={Pe{\~n}alol{\'e}n},
            postcode={7910000}, 
            state={Región Metropolitana},
            country={Chile}}
            
\author[labelValerio]{Valerio Dose}

\affiliation[labelValerio]{organization={Dipartimento di Ingegneria Informatica, Automatica e Gestionale, ``Sapienza'' Universit\`a  di Roma},%Department and Organization
            addressline={Via Ariosto 25}, 
            city={Roma},
            postcode={00185}, 
            country={Italy}}

\author[labelMarco]{Marco Scarsini}

\affiliation[labelMarco]{organization={Dipartimento di Economia e Finanza, Luiss University},%Department and Organization
            addressline={Viale Romania 32}, 
            city={Roma},
            postcode={00197}, 
            country={Italy}}

\begin{abstract}
We consider the behavior of the \acl{PoA} and  equilibrium flows in nonatomic multi-commodity routing games as a function of the traffic demand. 
We analyze their smoothness with a special attention to specific values of the demand at which the support of the Wardrop equilibrium exhibits a phase transition with an abrupt change in the set of optimal routes. 
Typically, when such a phase transition occurs, the \acl{PoA} function has a breakpoint, \ie is not differentiable. 
We prove that, if the demand varies proportionally across all commodities, then, at  a breakpoint,  the largest  left or right derivatives of the \acl{PoA} and of the social cost at equilibrium, are associated with the smaller equilibrium support.
This proves---under the assumption of proportional demand---a conjecture of \citet{OHaConWat:TRB2016}, who observed this behavior in simulations. 
We also provide counterexamples showing that this monotonicity of the one-sided derivatives may fail when the demand does not vary proportionally,  even if it moves along a straight line not passing through the origin.
\end{abstract}





\end{frontmatter}

%\linenumbers


% Paper body

%
% Section --------------------------------------
%

\section{Introduction}\label{se:intro}
Wardrop equilibria for nonatomic routing games provide a mathematical description of how traffic distributes over a network used by a large number of agents who do not coordinate or cooperate. 
Standard examples are road and telecommunication networks. 
This model is a special case of a nonatomic congestion game where the resources are the edges of a directed multigraph,
and each edge $\edge$ has an associated cost or delay 
$\cost_{\edge}(\load_{\edge})$ given by a non-decreasing and continuous  function  of the traffic load $\load_{\edge}$ on the edge.
Different types of users have different \ac{OD} pairs, typically called  commodities. 
A nonnegative traffic demand  is associated to each \ac{OD} pair, and users can choose different paths to go from their  origin to their destination.
In doing so, they generate a flow over the network.
A Wardrop equilibrium is a distribution of traffic which satisfies the demand on each \ac{OD} pair in a such a way that all the users of any given type only use paths of minimal cost.

Traffic demand is often difficult to predict as it is eminently a dynamic phenomenon, with variations between days of the week and across different periods within any given day. 
Moreover, traffic demand increases on a longer time scale as a result of the growth of the population and car-ownership. 
As a consequence, it is relevant to analyze the sensitivity of Wardrop equilibria under different demand scenarios to understand the extent to which this equilibrium notion  approximates real traffic patterns. 
Using this analysis, a planner can anticipate the evolution of equilibria and better design policy interventions and infrastructure modifications that the increasing traffic level requires.
These observations have motivated several authors to investigate how the Wardrop equilibrium is affected by variations in the traffic demands, in terms of its continuity, smoothness, and monotonicity properties. 
In  \cref{Sec:Related_Work} we present a brief overview of previous works that have addressed these questions.

On the other hand, Wardrop equilibria are known to be inefficient in the sense that there may exist other traffic distributions that produce a lower total  delay. 
Inefficiency of equilibria is usually measured via the \acfi{PoA}\acused{PoA}, \ie the ratio between the total equilibrium cost and the optimum total cost. 
This provides a benchmark for the maximal improvement that could be achieved in terms of social cost.
While the initial studies on the \ac{PoA} focused on finding tight bounds under worst case scenarios,  recent work has  turned to analyze the  behavior of the \ac{PoA} as a function of the traffic demand.

\begin{figure}[ht]
\setcounter{subfigure}{0}
\begin{subfigure}[Some network]
{
\begin{tikzpicture}
   \node[shape=circle,draw=black,line width=1pt,minimum size=0.5cm] (v1) at (-3,0)  { $\source$}; 
   \node[shape=circle,draw=black,line width=1pt,minimum size=0.5cm] (v2) at (0,1.3)  {$\vertex_{1}$}; 

   \node[shape=circle,draw=black,line width=1pt,minimum size=0.5cm] (v5) at (0,-1.3)  {$\vertex_{2}$}; 
   \node[shape=circle,draw=black,line width=1pt,minimum size=0.5cm] (v6) at (3,0)  {$\sink$}; 
    
   \draw[line width=1pt,->] (v1) to   node[midway,fill=white] {$\frac{\load}3$} (v2);
   \draw[line width=1pt,->] (v1) to   node[midway,fill=white] {$1$} (v5);
   \draw[line width=1pt,->] (v2) to   node[midway,fill=white] {$0$} (v5);
   
   \draw[line width=1pt,->] (v2) to   node[midway,fill=white] {$1$} (v6);

   \draw[line width=1pt,->] (v5) to   node[midway,fill=white] {$\load$} (v6);
   \draw[line width=1pt,->] (v1) to [bend right=90]  node[midway,fill=white] {$\frac 52+\load$} (v6); \draw[line width=1pt,->] (v1) to [bend left=90]  node[midway,fill=white] {$4+\frac{\load}2$} (v6);   
        
\end{tikzpicture}
\label{fig:graph_intro}
}
\end{subfigure}
\hspace{0.5cm}
\begin{subfigure}[\ac{PoA} for the network on the left]
{
\includegraphics[width=0.5\textwidth]{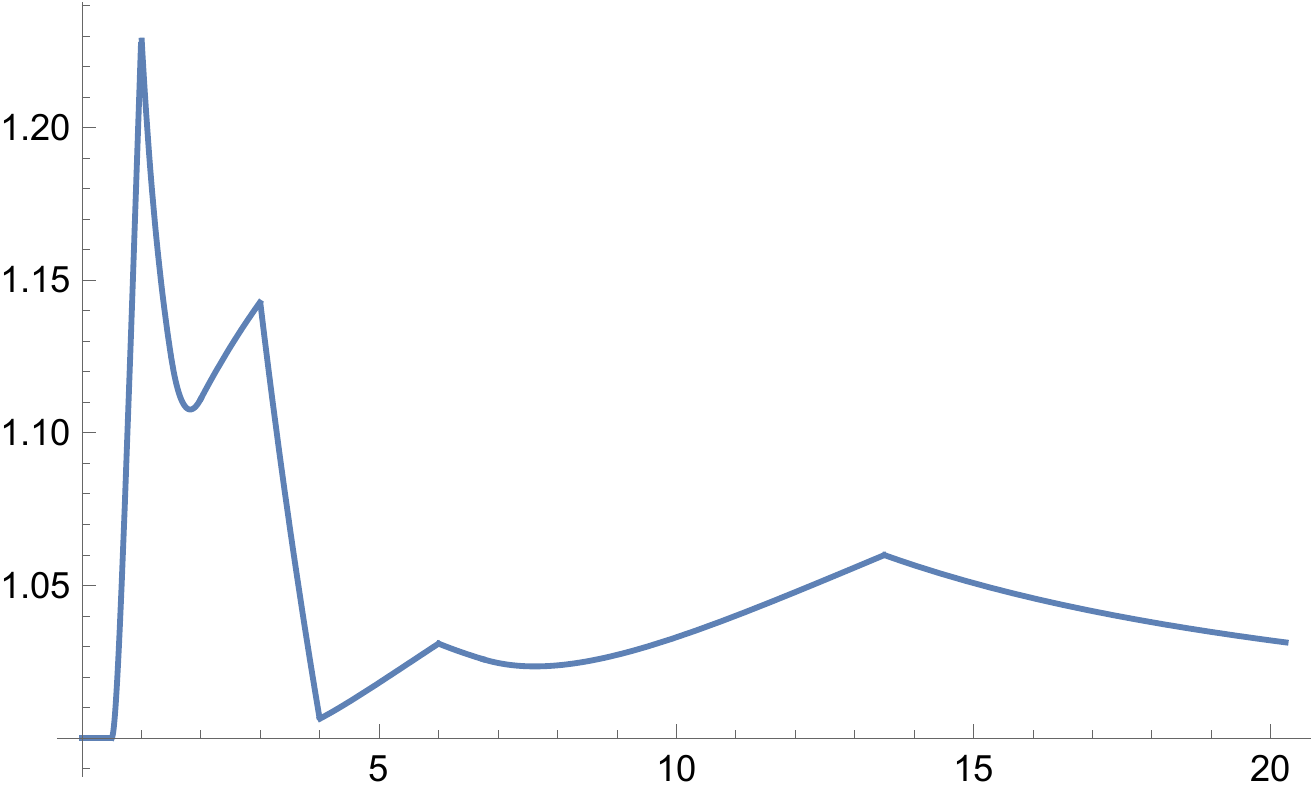}
\label{fig:PoA_intro}}
\end{subfigure}
\caption{ A single-commodity example of the behavior of the \ac{PoA} as a function of the demand.
This function is nonsmooth at demands $1$, $3$, $4$, $6$ and $27/2$.  
Between these breakpoints the set of optimum paths remains stable and the \ac{PoA} behaves smoothly.
At demand $4$ the set of optimum paths exhibits a contraction and the derivative jumps up. 
At all the other breakpoints the set of  optimum paths undergoes an expansion and the derivative jumps down.
}
\label{fig:intro}
\end{figure}

\cref{fig:intro} shows a typical profile of the \ac{PoA} as the demand varies: it is close to $1$ when the demand is either very small or very large, and it oscillates in the middle range with occasional sharp kinks and smooth critical points at local minima. As reported by \cite{YouGasJeo:PRL2008}, this type of  profile
is also observed in real networks such as Boston, London, and New York, where in addition the $\PoA$ values remain far from the worst case estimates. The fact that the \ac{PoA} is close to $1$  for both low and high demands was formally established under fairly general assumptions in \cite{ColComSca:TOCS2019,ColComMerSca:OR2020} and \cite{WuMohCheXu:OR2021}.
For the intermediate range, \citet{OHaConWat:TRB2016} observed that the kinks are associated with a sudden change in the set of optimum paths at equilibrium, and stated a conjecture that we analyze in this paper. A partial support for this was given by \citet{ComDosScar:MathPRog2021}:
in single-commodity routing games with affine cost functions, and in demand intervals where the set of optimum paths at equilibrium does not change, the \ac{PoA} is differentiable and exhibits at most one critical point, which  must be a local minimum. 
In this paper we further explore these questions for multi-commodity networks and more general cost functions.

\subsection{Our Results}
We consider nonatomic routing games and  analyze the behavior of the equilibrium flows, the social cost, and the \ac{PoA}, as the traffic demands vary on potentially multiple  \ac{OD} pairs.   The whole paper deals with \emph{constrained} routing games in which every \ac{OD} pair is allowed to use a restricted set of paths and not necessarily all the paths connecting its origin and destination. 
The main takeaway of our study is that when the demands vary with constant proportions across all \ac{OD} pairs, the behavior of the \ac{PoA} is analogue to what is observed in the single-commodity  case. 
However, this is not what happens if the demand varies without keeping constant proportions among the different \ac{OD} pairs.

Specifically, we begin by generalizing to multiple \ac{OD} pairs some results on the continuity and differentiability of equilibria as a function of the demands. 
We show that if the cost functions are $\mathcal C^{1}$ with strictly positive derivative, and the demand varies along a smooth curve  parameterized by a single variable $\var$, then the equilibrium loads on each edge and the \ac{PoA} are also $\mathcal C^{1}$ as functions of $\var$, as long as the  demand curve stays in a region where the set of minimum-cost paths does not change. 

Typically---although not always---this smoothness fails at points where the set of  optimum paths  changes
and the equilibrium exhibits a phase transition. 
We call such points \emph{breakpoints}.
For networks with $\mathcal C^{1}$ cost functions having strictly positive derivatives and proportionally varying demands, we prove a conjecture stated in \citet{OHaConWat:TRB2016}, showing that at a breakpoint the largest between the left and right derivatives of the \ac{PoA} corresponds to the smaller set of  optimum paths (see the example in \cref{fig:intro}).
We also show with an example that the full conjecture, as originally stated in \citet{OHaConWat:TRB2016}, 
might fail when the demands do not vary proportionally, even if these demands move along an affine straight line in the space of demands.

\subsection{Related Work}\label{Sec:Related_Work}

The definition of equilibrium in nonatomic routing games that we adopt in this paper is due to  \citet{War:PICE1952}.
The characterization of Wardrop equilibria as solutions of a convex optimization problem was first described in \citet{BecMcGWin:Yale1956}; and the first attempts to explicitly compute such equilibria were due to \citet{Tom:OR1966} for affine costs, and to \citet{DafSpa:JRNBSB1969} for general convex costs. 
For surveys on the topic the reader is referred to  \citet{FloHea:HORMS1995} and \citet{CorSti:EORMS2011}.

Several authors have considered how the equilibrium flows and costs vary with the traffic demand. 
\citet{Hal:TS1978} proved that the equilibrium cost of any given \ac{OD} pair is an increasing function of the amount of traffic on that \ac{OD} pair. 
On the other hand \citet{Fis:TRB1979} presented an example showing that the total social cost can decrease even if the total demand increases. 
Further analysis of this  paradoxical phenomenon can be found in \citet{DafNag:MP1984}. 
Concerning the smoothness of equilibrium flows, \citet{Pat:TS2004} gave  necessary and sufficient conditions for the directional differentiability of equilibria, and 
\citet{JosPat:TRB2007} proved the directional differentiability of equilibrium costs, but observed that this might fail for the equilibrium edge loads.
Concerning  the set of optimum paths, \citet{EngFraOlb:TCS2010} showed examples of single \ac{OD} routing games where an arbitrarily small increase in traffic demand generates a complete change of the set of paths used at equilibrium, although the change on the edge loads remains small. 
Moreover, for polynomial costs of degree at most $d$,  if the demand increases by $\varepsilon$ then the equilibrium costs increase at most by a multiplicative factor $(1+\varepsilon)^{d}$. 
The extension of this latter result to multiple \ac{OD} pairs can be found in \citet{TakKwo:OL2020}. 
More recently, an analysis of instances where there exists a Wardrop equilibrium with loads that are nondecreasing for each edge was carried out in \citet{ComDosSca:EJOR2024}, whereas a study on conditions that allow Braess' paradox to happen at specific demands was performed by \citet{VerChe:arXiv2023}. 
In this last work the authors  studied the slope of the equilibrium costs, in the setting of networks with affine cost functions; these results are similar to some theorems proved here.

\citet{Pig:Macmillan1920} was probably the first author who studied  the inefficiency of equilibria in nonatomic routing games. 
The formal definition of  \ac{PoA} to measure this inefficiency is due to \citet{KouPap:STACS1999}, and acquired its name in \citet{Pap:ACMSTC2001}.
Most of the early literature on the \ac{PoA} concentrated on establishing sharp bounds for the \ac{PoA} for different classes of games, such as congestion games and routing games.
In a landmark paper, \citet{RouTar:JACM2002}  proved that in every nonatomic congestion game with affine costs the \ac{PoA} is bounded above by $4/3$ and showed that this bound is sharp.
This result was generalized to polynomial cost functions of maximum degree $d$  in \citet{Rou:JCSS2003}, showing that the \ac{PoA} grows as $\Theta(d/\log d)$. 
Other results of this type can be found in \citet{DumGai:INE2006}, who focused on cost functions that are sums of monomials whose degrees are in a specified range, and in
\citet{RouTar:GEB2004}, who dealt with the class of differentiable cost functions $\cost(\load)$
such that $\load\, \cost(\load)$ is convex.
Less regular cost functions and different notions of social cost were studied in \citet{CorSchSti:MOR2004,CorSchSti:OR2007,CorSchSti:GEB2008}.

Several papers addressed the computation of the \ac{PoA} in real networks. 
\citet{MonBenPil:WINE2018} analyzed data from a large sample of Singaporean students who commute to go to school,
observing a \ac{PoA} which is much smaller than the theoretical worst case bounds. Also,
\citet{YouGasJeo:PRL2008,YouGasJeo:PRL2009} studied the \ac{PoA} in the networks of  Boston, London, and New York when all the OD traffic demands are scaled by the same factor, observing again values of the \ac{PoA} consistently smaller than the worst case bounds.
The behavior of the  \ac{PoA}  in these three cities shows a common pattern: it is close to $1$ both for small and large demands, and it oscillates in the middle range with sharp kinks at the local maxima and smooth critical points at the local minima.
\citet{OHaConWat:TRB2016} noted that these kinks arise when the set of optimum paths at equilibrium undergoes an expansion or a contraction, and stated 
the conjecture that we analyze later in this paper. 

Recent efforts  attempted to mathematically explain the empirical behavior of the \ac{PoA} observed in the studies mentioned above. 
An algorithm for computing Wardrop equilibria as a function of the traffic demand was given in \citet{KliWar:SODA2019,KliWar:MOR2022}, in the case of piecewise linear cost functions.
The efficiency of equilibria when the demand is close to zero or infinity (light and heavy traffic) was analyzed in \citet{ColComSca:TOCS2019,ColComMerSca:OR2020} and \citet{WuMohCheXu:OR2021}.  
\citet{ColComSca:TOCS2019} considered the case of single \ac{OD} parallel networks showing that in heavy traffic the \ac{PoA} converges to $1$ when the cost functions are regularly varying. 
These results were extended to general networks in \citet{ColComMerSca:OR2020}, 
considering also the behavior of \ac{PoA} in the light traffic regime. 
The case of heavy traffic was  treated with different techniques by  \citet{WuMohCheXu:OR2021}. 
The study of the intermediate range of demands, when the traffic is neither  light nor heavy, was the main objective of \citet{ComDosScar:MathPRog2021}, who proved that for affine cost functions the shape of the \ac{PoA} function is  the one observed in the empirical studies. 
Finally, we mention 
\citet{WuMoh:MOR2023}, who established the continuity of the \ac{PoA} in non atomic routing games, as a function of various parameters including not only the demands, but also the cost functions.

\subsection{Organization of the paper} 
\label{suse:organization}
\cref{se:network-games}
recalls the model to be studied, and fixes the notations used thereafter. \cref{se:regularity} presents some regularity 
and smoothness results for the equilibrium loads, the social cost at equilibrium, the optimum social cost, and the \acl{PoA}.
\cref{se:singleparameter} discusses the behavior of these quantities around breakpoints, where smoothness is typically lost.
\cref{se:proof-main-theorem} contains the proof of the main result.
Some other technical proofs are presented in \cref{se:supplementary}, whereas 
\cref{se:symbols} contains a list of the symbols used in the paper.

%
% Section --------------------------------------
%

\section{Network games with variable demand}
\label{se:network-games}

This section introduces the model and notations used throughout the paper. 
We consider routing games with multiple \ac{OD} pairs,
and we study their equilibria  as functions of the multivariate demand, 
concentrating on the case where the demand vector varies 
 along a smooth curve in the space of demands. 
In particular, we study the simplest and  natural case where  the demands are scaled by a common factor so that they change proportionally across \ac{OD} pairs.

Let $\graph\coloneqq (\vertices,\edges)$ be a directed multigraph with vertex set $\vertices$ and  edge set $\edges$, where each edge $\edge\in\edges$ has a  continuous and nondecreasing \emph{cost function} $\cost_{\edge}\colon[0,+\infty)\rightarrow[0,+\infty)$, with $\cost_{\edge}(\load_{\edge})$ representing the travel time of traversing the edge $\edge$ when the load on that edge is $\load_{\edge}$.
An \acfi{OD}\acused{OF}  $\odpair$ is defined by a triple $(\source^{(\odpair)},\sink^{(\odpair)},\routes^{(\odpair)})$, where $\source^{(\odpair)}\in\vertices$ is the origin vertex, $\sink^{(\odpair)}\in\vertices$ is the destination vertex, and $\routes^{(\odpair)}$
is a subset of simple paths from  $\source^{(\odpair)}$ to $\sink^{(\odpair)}$. 
The symbol $\ODpairs$ denotes the set of \ac{OD}'s.
The sets $\routes^{(\odpair)}$ are assumed nonempty but could be smaller than the sets of all paths from origin to destination in the graph $\graph$. 
In order to simplify our notations and formulas, we assume that these sets are pairwise disjoint so that no path is common to two different \ac{OD}'s.
In particular, when considering the flow of a path, we do not need to specify to which \ac{OD} the traffic flow belongs. We observe that our results do not depend on this simplifying assumption (see  \cref{rm:assumption-disjoint}). 
The set of all feasible routes is the disjoint union of the sets $\routes^{(\odpair)}$ for  $\odpair\in\ODpairs$, denoted by
\begin{equation}
\label{eq:paths}
\routes \coloneqq \union_{\odpair\in\ODpairs} \routes^{(\odpair)}.   
\end{equation}

\begin{definition}
\label{de:nonatomic-routing-game}
A \emph{routing game structure} is a tuple $(\graph,\costprof,\ODpairs)$, where  $\graph$ is a directed multigraph, $\costprof$ is the vector of nondecreasing and continuous edge costs, and $\ODpairs$ is the set of \ac{OD} pairs.
\end{definition} 

A \emph{traffic demand} for the routing game structure $(\graph,\costprof,\ODpairs)$ is a vector $\rateprof=(\rate^{(\odpair)})_{\odpair\in\ODpairs}\in\Rpos^{\ODpairs}$. 
Every such $\rateprof$  determines an associated set of \emph{feasible flows} $\flowprof = \parens*{\flow_{\route}}_{\route\in\routes}\in\Rpos^{\routes}$ satisfying
\begin{equation}
\label{eq:flows} 
\sum_{\route\in\routes^{(\odpair)}}\flow_{\route}=\rate^{(\odpair)}\text{ for all }\odpair\in\ODpairs.  
\end{equation}
These flows induce in turn a \emph{load profile} $\loadprof=(\load_{\edge})_{\edge\in\edges}$, where $\load_{\edge}$ represents the aggregate traffic over the edge $\edge\in\edges$ given by
\begin{equation}\label{eq:loads}
\load_{\edge}=\sum_{\route\in\routes}\edgepath_{\edge\route} \flow_{\route}\text{ for all }\edge\in\edges,
\end{equation}
with $\edgepath_{\edge\route}=1$ if $\edge\in\route$ and  $\edgepath_{\edge\route}=0$ otherwise.
More concisely, \eqref{eq:flows} and \eqref{eq:loads} can be written in matrix form as 
\begin{equation}
\label{eq:matricesConstraints}
 \sumflow \flowprof=\rateprof,\quad \loadprof=\edgepaths \flowprof.
\end{equation}
We write $\flows_{\rateprof}$ for the set of pairs $(\flowprof,\loadprof)\in\Rpos^{\routes}\times\Rpos^{\edges}$ satisfying these flow/load conservation equations. We also write 
$\loads_{\rateprof}$ for the projection of $\flows_{\rateprof}$
onto the the space of load variables
$\loadprof$.

With a slight abuse of notation,
the total travel time on a path $\route\in\routes$ is denoted

\begin{equation}\label{eq:path_cost}
\cost_{\route}(\loadprof)\coloneqq \sum_{\edge\in\route}\cost_{\edge}(\load_{\edge}).
\end{equation}
The total delay induced by  $(\flowprof,\loadprof)$ is called the \emph{social cost} and is denoted by
\begin{equation}
\label{eq:social-cost}  
\SC(\flowprof,\loadprof):=
\sum_{\route\in\routes}\flow_{\route}\cdot\cost_{\route}(\loadprof)=
\sum_{\edge\in\edges}\load_{\edge}\cdot\cost_{\edge}(\load_{\edge}).
\end{equation}

%
% Subsection --------------------------------------
%

\subsection{Equilibria and price of anarchy }
\label{sec:active-regimes-and-breakpoints}

Let $(\graph,\costprof,\ODpairs)$ be  a routing game structure. 
For each demand vector $\rateprof\in\Rpos^{\ODpairs}$, we obtain a classical nonatomic routing game. 
We recall that a Wardrop equilibrium is a feasible flow-load pair $(\eq{\flowprof},\eq{\loadprof})\in \flows_{\rateprof}$ for which there exists  $\eqcostprof\coloneqq (\eqcost^{(\odpair)})_{\odpair\in\ODpairs}\in\reals^{\ODpairs}$ such that
\begin{equation}
\label{eq:Wardrop}
\begin{cases}
\cost_{\route}(\eq{\loadprof})=\eqcost^{(\odpair)} & \text{for every }\odpair\in\ODpairs\text{ and all }\route\in\routes^{(\odpair)}\text{ with }\eq{\flow}_{\route}>0,\\
\cost_{\route}(\eq{\loadprof})\ge\eqcost^{(\odpair)} & \text{for every }\odpair\in\ODpairs\text{ and all }\route\in\routes^{(\odpair)}\text{ with }\eq{\flow}_{\route}=0.
\end{cases}
\end{equation}
The quantity $\eqcost^{(\odpair)}$ is called the \emph{equilibrium cost} associated to the \ac{OD} pair $\odpair\in\ODpairs$.

From \citet{BecMcGWin:Yale1956} we know that the set of all  equilibrium load profiles $\eq{\loadprof}$ coincides
with the set of optimum solutions of the minimization problem
\begin{equation}
\label{eq:min-load}    
\min_{\loadprof\in\loads_{\rateprof}}\potential(\loadprof),
\end{equation}
where $\potential(\loadprof):=\sum_{\edge\in\edges}\Cost_{\edge}(\load_{\edge})$ with $\Cost_{\edge}(\load_{\edge})\coloneqq \int_0^{\load_{\edge}}\cost_{\edge}(z)\diff z$. 
Since the cost functions $\cost_{\edge}$ are continuous and nondecreasing, the objective function $\potential(\argdot)$ is  $\mathcal C^{1}$ and  convex. 
Hence,  the minimum is attained and therefore for every $\rateprof\in\Rpos^{\ODpairs}$ there exists at least one equilibrium.

An equilibrium $(\eq{\flowprof},\eq{\loadprof})\in\flows_{\rateprof}$ induces equilibrium edge costs $\eqcostedge_{\edge}\coloneqq \cost_{\edge}(\eq{\load}_{\edge})$. 
In \citet{Fuk:TRB1984}, the equilibrium edge costs were shown to be optimum solutions of the strictly convex dual program 
\begin{equation}\label{eq:Fukushima}
\min_{\boldsymbol{\tau}}\sum_{\edge\in\edges}\Cost_{\edge}^{*}(\tau_{\edge}) -
\sum_{\odpair\in\ODpairs}\parens*{\rate^{(\odpair)}\,\min_{\route\in\routes^{(\odpair)}}\sum_{\edge\in\route}\tau_{\edge}},
\end{equation}
where $\Cost_{\edge}^{*}(\argdot)$ is the Fenchel conjugate of $\Cost_{\edge}(\argdot)$, which is strictly convex. 
Thus, the edge costs $\eqcostedge_{\edge}$ at equilibrium are uniquely defined for each demand vector $\rateprof$, and are the same in every equilibrium load profile $\eq{\loadprof}$. This defines maps
$\rateprof\mapsto \eqcostedge_{\edge}(\rateprof)$
that assign these equilibrium edge costs 
to each demand vector $\rateprof$.
It follows that, at equilibrium, the cost $\eqcostroute_{\route}(\rateprof)=\sum_{\edge\in\routes}\eqcostedge_{\edge}(\rateprof)$ of a path $\route\in\routes$, as well as the  equilibrium costs $\eqcost^{(\odpair)}(\rateprof)=\min_{\route\in\routes^{(\odpair)}}\eqcostroute_{\route}(\rateprof)$ for each $\odpair\in\ODpairs$, are also uniquely defined for each $\rateprof$ and do not depend on the specific equilibrium  $(\eq{\flowprof},\eq{\loadprof})$ that we might consider.

Moreover, when the edge costs are strictly increasing, the equilibrium loads are also uniquely determined by $\eq{\load}_{\edge}(\rateprof)=\cost_{\edge}^{-1}(\eqcostedge_{\edge}(\rateprof))$.
For later reference, we introduce the concept of active regime, which associates to each $\rateprof$ the set of minimum-cost paths at equilibrium.
\begin{definition}
\label{de:regime}
Consider the routing game structure $(\graph,\costprof,\ODpairs)$.
\begin{enumerate}[label=\textup{(\alph*)}, ref=\textup{(\alph*)}] 
\item 
\label{it:regime}
A subset  $\regime\subset\routes$ is called a \emph{regime} if $\regime\cap\routes^{(\odpair)}\ne\varnothing$ for every $\odpair\in\ODpairs$.

\item 
\label{it:active-regime}
The \emph{active regime} at demand $\rateprof\in\Rpos^{\ODpairs}$ is the set
\begin{equation}
\label{eq:active-regime} 
\activ{\routes}(\rateprof)=\bigcup_{\odpair\in\ODpairs}\braces*{\route\in\routes^{(\odpair)}\colon \eqcostroute_{\route}(\rateprof)=\eqcost^{(\odpair)}(\rateprof)}.
\end{equation}
\end{enumerate}
\end{definition}

All of this allows us to define the \emph{equilibrium social cost} with demand $\rateprof\in\Rpos^{\ODpairs}$ as
\begin{equation}
\label{eq:equilibrium-social-cost}  
\eq{\SC}(\rateprof)\coloneqq \sum_{\odpair\in\ODpairs}\rate^{(\odpair)}\eqcost^{(\odpair)}(\rateprof)
=\sum_{\route\in\routes}\eq{\flow}_{\route}\,\cost_{\route}(\eq{\flowprof})=\sum_{\edge\in\edges}\eq{\load}_{\edge}\,\cost_{\edge}(\eq{\load}_{\edge}),
\end{equation}
where $(\eq{\flowprof},\eq{\loadprof})\in\flows_{\rateprof}$ is  any equilibrium flow.

Also, the \emph{optimum social cost} for a demand vector $\rateprof\in\Rpos^{\ODpairs}$ is defined as 
\begin{equation}
\label{eq:optimum-social-cost}    
\opt{\SC}(\rateprof)
\coloneqq \min_{(\flowprof,\loadprof)\in\flows_{\rateprof}}\sum_{\route\in\routes}\flow_{\route}\,\cost_{\route}(\loadprof)
=\min_{\loadprof\in\loads_{\rateprof}}\sum_{\edge\in\edges}\load_{\edge}\,\cost_{\edge}(\load_{\edge}),
\end{equation}
and the \emph{\acl{PoA}} at demand $\rateprof\in\Rpos^{\ODpairs}$ is 
\begin{equation}
\label{eq:poa} 
\PoA(\rateprof)\coloneqq \frac{\eq{\SC}(\rateprof)}{\opt{\SC}(\rateprof)}\ge 1.
\end{equation}

%
% Section --------------------------------------
%

\section{Continuity and smoothness of equilibria and \ac{PoA}.}
\label{se:regularity}

In what follows we study the behavior of the equilibrium,  social cost, and price-of-anarchy, when the demand varies on a smooth curve parameterized by a scalar variable. For instance,  a linear demand function $\rateprof(\var)=\var\,\odrateprof$ for a fixed vector $\odrateprof\in \Rpos^{\nOD}$ represents a situation where the demands vary proportionally across different \ac{OD} pairs.

\begin{definition}
A \emph{nonatomic routing game with demand function $\rateprof(\argdot)$} is given by  a tuple $(\graph,\costprof,\ODpairs,\rateprof(\argdot))$,  where the demand is a function $\rateprof \colon [0,\infty) \to \Rpos^{\nOD}$.
\end{definition}

\begin{theorem}
\label{th:differentiability-on-curve}
Let $(\graph,\costprof,\ODpairs,\rateprof(\argdot))$ be a nonatomic routing game with a continuously differentiable demand function $\rateprof(\argdot)$. 
Let the cost functions $\cost_{\edge}$ be $\mathcal C^{1}$ with strictly positive derivative,
and suppose that the active regime $\activ{\routes}(\rateprof(\var))$ is constant on a neighborhood of $\var^0\in\Rpos$. Then:
\begin{enumerate}[label=\textup{(\alph*)}, ref=\textup{(\alph*)}] 
\item
\label{th:differentiability-on-curve-a}
The equilibrium  load map $\var\mapsto \eq{\loadprof}(\rateprof(\var))$ is continuously differentiable on a neighborhood of $\var^0$. 
In particular the equilibrium costs $\parens*{\eqcost^{(\odpair)}(\rateprof(\var))}_{\odpair\in\ODpairs}$ and the social cost at equilibrium $\eq{\SC}(\rateprof(\var))$ are of class $\mathcal{C}^1$ on a neighborhood of $\var^0$.
\item
\label{prop:PoA-regularity}
If in addition the maps $\load_{\edge}\mapsto\load_{\edge}\cost_{\edge}(\load_{\edge})$ are convex, then minimum social cost $\opt{\SC}(\rateprof(\var))$ and the \acl{PoA} $\PoA(\rateprof(\var))$ are of class $\mathcal{C}^1$ on a neighborhood of $\var^0$.
\end{enumerate}
\end{theorem}

The simplest situation
where the active regime  $\activ{\routes}(\rate(\var))$ is constant on a neighborhood of $\var^0$ is when every path of minimal cost at demand $\rate(\var^0)$ carries a strictly positive flow. In this case, by continuity of the equilibrium costs, we are assured that the active regime will not change after small changes in the demands.

The proof of \cref{th:differentiability-on-curve} exploits the regularity properties of the equilibrium costs and optimum social cost stated below in \cref{lem:continuity,lem:optimum-social-cost-smooth}.
\cref{lem:continuity}  extends \citet[proposition 3.1]{ComDosScar:MathPRog2021} to the multi-OD setting, as well as \citet{Hal:TS1978}, who considered the case of strictly increasing costs. The proof of \cref{lem:continuity} can be found in \citet[proposition~1]{ComDosSca:EJOR2024}, while  \cref{lem:optimum-social-cost-smooth}  is proved in \cref{se:supplementary}.

\begin{lemma}\label{lem:continuity}
Let $(\graph,\costprof,\ODpairs)$ be a routing game structure. 
Then,  the equilibrium edge costs $\rateprof\mapsto\eqcostedge_{\edge}(\rateprof)$ are uniquely defined and continuous. 
Moreover, the equilibrium costs $\rateprof\mapsto\eqcost^{(\odpair)}(\rateprof)$ are continuous.  
\end{lemma}

\begin{lemma}
\label{lem:optimum-social-cost-smooth}
Let $(\graph,\costprof,\ODpairs)$ be a routing game structure. 
If the cost functions $\cost_{\edge}$ are $\mathcal C^{1}$ and the functions $\load_{\edge}\mapsto\load_{\edge}\cost_{\edge}(\load_{\edge})$ are convex, then the optimum social cost function $\rateprof\mapsto\opt{\SC}(\rateprof)$ is convex and $\mathcal C^{1}$ everywhere.
\end{lemma}

In addition to these lemmas, our proof of \cref{th:differentiability-on-curve} involves the analysis of some auxiliary minimization problems without sign constraints on the variables. 
In order to properly formulate these problems, we extend the edge cost functions to the whole $\mathbb{R}$ in such a way that they remain $\mathcal{C}^1$ with strictly positive derivatives and with $\lim_{\load\to-\infty}\cost_{\edge}(\load)<0$ for all edges $\edge\in\edges$. 
Then, for each fixed regime $\regime$, we consider the auxiliary problems:
\begin{equation}
\label{eq:Beckmann-fixed-regime-relaxed}
\min_{\loadprof\in\loads_{\rateprof}^{\regime}}\potential(\loadprof)
\tag{$\prim_{\rateprof}^{\regime}$}
\end{equation}
where the feasible set $\loads^\regime_{\rateprof}$
comprises all \emph{signed} load vectors
$\loadprof \in \mathbb{R}^{\edges}$ induced by some flow  $\flowprof\in\mathbb{R}^{\routes}$ with support in $\regime$, that is
\begin{equation}
\label{eq:X-mu-R}    
\loads^\regime_{\rateprof} =
\braces*{\loadprof \in \mathbb{R}^{\edges} \colon (\exists \flowprof\in\mathbb{R}^\routes) \text{ s.t. } \sumflow \flowprof=\rateprof,\;\loadprof=\edgepaths \flowprof,\text{ and }\flow_{\route}=0\text{ for all } \route\notin\regime}.
\end{equation}

 \cref{lem:differentiability-of-fixed-regime-possibly-negative-flows}~\ref{it:lem:differentiability-of-fixed-regime-possibly-negative-flows-b} in \cref{se:supplementary} shows that
when the costs  $\cost_{\edge}$ are $\mathcal C^{1}$ with strictly positive derivative, then 
\eqref{eq:Beckmann-fixed-regime-relaxed} has a unique optimal solution $\loadprof_{\regime}(\rateprof)$ and the map $\rateprof\rightarrow\loadprof_{\regime}(\rateprof)$ is of class $\mathcal C^{1}$. Moreover, 
\cref{lem:differentiability-of-fixed-regime-possibly-negative-flows}~\ref{it:lem:differentiability-of-fixed-regime-possibly-negative-flows-a} provides explicit expressions for the Lagrange multipliers, whose uniqueness is used in the proof of \cref{th:WatlingSC}~\ref{it:th:WatlingSC-2}.

Let us insist that in \eqref{eq:Beckmann-fixed-regime-relaxed} the regime $\regime$ is fixed and
 the variables have no sign constraints, so this problem is a relaxation of \eqref{eq:min-load} and we might have $\flow_\route<0$ 
 for some paths $\route\in\routes^{(\odpair)}\cap\regime$. Moreover, 
 although for every optimal solution $\loadprof$  of \eqref{eq:Beckmann-fixed-regime-relaxed} and each OD pair $\odpair$ the costs of all paths $\route\in\routes^{(\odpair)}\cap\regime$ turn out to be equal to some common value $\multod^{(\odpair)}$, it can happen that  some paths $\route\in\routes^{(\odpair)}\setminus\regime$ may have a 
 smaller cost $\cost_{\route}(\loadprof)<\multod^{(\odpair)}$.
 However, if for some demand vector $\rateprof$ we have an optimal solution with
$\flow_\route\geq 0$ for all $\route\in\routes^{(\odpair)}\cap\regime$ and $\cost_{\route}(\loadprof)\geq\multod^{(\odpair)}$ for all  $\route\in\routes^{(\odpair)}\setminus\regime$, then this is also an optimal solution of the 
original problem \eqref{eq:min-load} and therefore it corresponds to a Wardrop equilibrium.

\begin{proof}[Proof of \cref{th:differentiability-on-curve}]
\ref{it:lem:differentiability-of-fixed-regime-possibly-negative-flows-a} A sufficient condition for
a point $\loadprof\in\loads^{\regime}_{\rateprof}$ to be an optimal solution
of \eqref{eq:Beckmann-fixed-regime-relaxed}, is the existence of multipliers $(\multodprof,\multedgeprof,\multpathprof)$ such that
\begin{align*}
\cost_{\edge}(\load_{\edge})
&=\multedge_{\edge} &&\forall\edge\in\edges,\\
\multod^{(\odpair)}
&=\sum_{\edge\in\route}\multedge_{\edge} &&\forall\odpair\in\ODpairs\;\forall\route\in\routes^{(\odpair)}\cap\regime,\\
\multod^{(\odpair)}
&=\sum_{\edge\in\route}\multedge_{\edge}-\multpath_{\route} &&\forall\odpair\in\ODpairs\;\forall\route\in\routes^{(\odpair}\setminus\regime.
\end{align*}
For any demand $\rateprof$ the  equilibrium load vector $\loadprof^*(\rateprof)$ satisfies these conditions for
 $\regime=\activ{\routes}(\rateprof)$, with $\multedge_{\edge}=\eqcostedge_{\edge}=\cost_{\edge}(\eq\load_{\edge})$, $\multod^{(\odpair)}=\eqcost^{(\odpair)}(\rateprof)$,
and $\multpath_{\route}=\sum_{\edge\in\route}\eqcostedge_{\edge}-\eqcost^{(\odpair)}(\rateprof)$ for  $\route\not\in\regime$. 
This shows that the equilibrium load $\loadprof^{*}(\rateprof)$ coincides with the unique optimum solution $\loadprof_{\regime}(\rateprof)$ of \eqref{eq:Beckmann-fixed-regime-relaxed} for $\regime=\activ{\routes}(\rateprof)$.
Now, by assumption $\activ{\routes}(\rateprof(\var))\equiv\regime$ is constant for $\var$ near $\var_0$, so that $\eq{\loadprof}(\rateprof(\var))=\loadprof_\regime(\rateprof(\var))$, which is
continuously differentiable as a composition of the $\mathcal{C}^{1}$ demand function $\var\mapsto\rateprof(\var)$ and the map $\rateprof\mapsto\loadprof_{\regime}(\rateprof)$,
which is also $\mathcal{C}^{1}$ in view of 
 \cref{lem:differentiability-of-fixed-regime-possibly-negative-flows}\!
 \ref{it:lem:differentiability-of-fixed-regime-possibly-negative-flows-b}.

\medskip
\noindent
\ref{it:lem:differentiability-of-fixed-regime-possibly-negative-flows-b} The smoothness of 
$\var\mapsto\opt{\SC}(\rateprof(\var))$
is a direct consequence of \cref{lem:optimum-social-cost-smooth}. This, combined with 
part (a), implies the smoothness 
of $\var\mapsto\PoA(\rateprof(\var))=\eq{\SC}(\rateprof(\var))/\opt{\SC}(\rateprof(\var))$.
\end{proof}

\begin{remark}
\label{re:strictly-pos}
\begin{enumerate}[label={\textup{(\alph*)}}, ref={\textup{(\alph*)}}]
\item 
Observe that \cref{th:differentiability-on-curve} requires the active regime $\activ{\routes}(\rateprof(\var))$ to be constant near $\var^0\in\Rpos$ but only along the demand curve, and not necessarily in a neighborhood of $\rateprof(\var^0)$ in the ambient space $\Rpos^{\ODpairs}$. 
This covers some special situations in which the curve $\rateprof(\argdot)$ may slide along the boundary between two regions with different active regimes. 

\item
Invoking \cref{rm:differentiability-with-cycles-condition} in the \cref{se:supplementary}, the strict positivity of the cost derivatives in \cref{th:differentiability-on-curve} can be slightly weakened by assuming that the edge costs $\cost_{\edge}(\argdot)$ are just strictly increasing and the set of edges $\edge$ with $\cost_{\edge}'(\load_{\edge}(\rateprof(\var)))=0$ do not contain undirected cycles.
\end{enumerate}
\end{remark}

The following examples show that the assumptions in \cref{th:differentiability-on-curve} are close to minimal: the presence of a single edge with a  non-convex cost whose derivative vanishes at a single point, may give rise to non-smooth equilibrium loads, even if the set of optimum paths is constant around the reference point. 
In \cref{ex:nondifferentiable-flow-with-no-breakpoint,ex:non-differentiable-2} below, the point where the cost functions have zero derivative is not the origin. 
Whether \cref{th:differentiability-on-curve} can be extended to cost functions that exhibit a zero derivative only at the origin remains an open problem. 
This would be of major interest since this feature is shared by  \ac{BPR} cost functions.

\begin{example}
\label{ex:nondifferentiable-flow-with-no-breakpoint}

In this example with strictly increasing $\mathcal C^{1}$ cost functions, the equilibrium flow is nondifferentiable at some point which is not a breakpoint. Consider a two link parallel network where the cost functions (plotted in \cref{fig:non_differentiable_1}) are
\begin{align}
\label{eq:c1}
\cost_1(\load)
&=\begin{cases}
-(\load-1)^2+1&\text{if }\load\le 1,\\
(\load-1)^2+1&\text{if }\load> 1,
\end{cases}\\
\label{eq:c2}
\cost_2(\load)
&=\begin{cases}
-(\load-1)^2+1&\text{if }\load\le 1,\\
2(\load-1)^2+1&\text{if }\load> 1.
\end{cases}\
\end{align}

\begin{figure}
\begin{center}
\includegraphics[width=7cm]{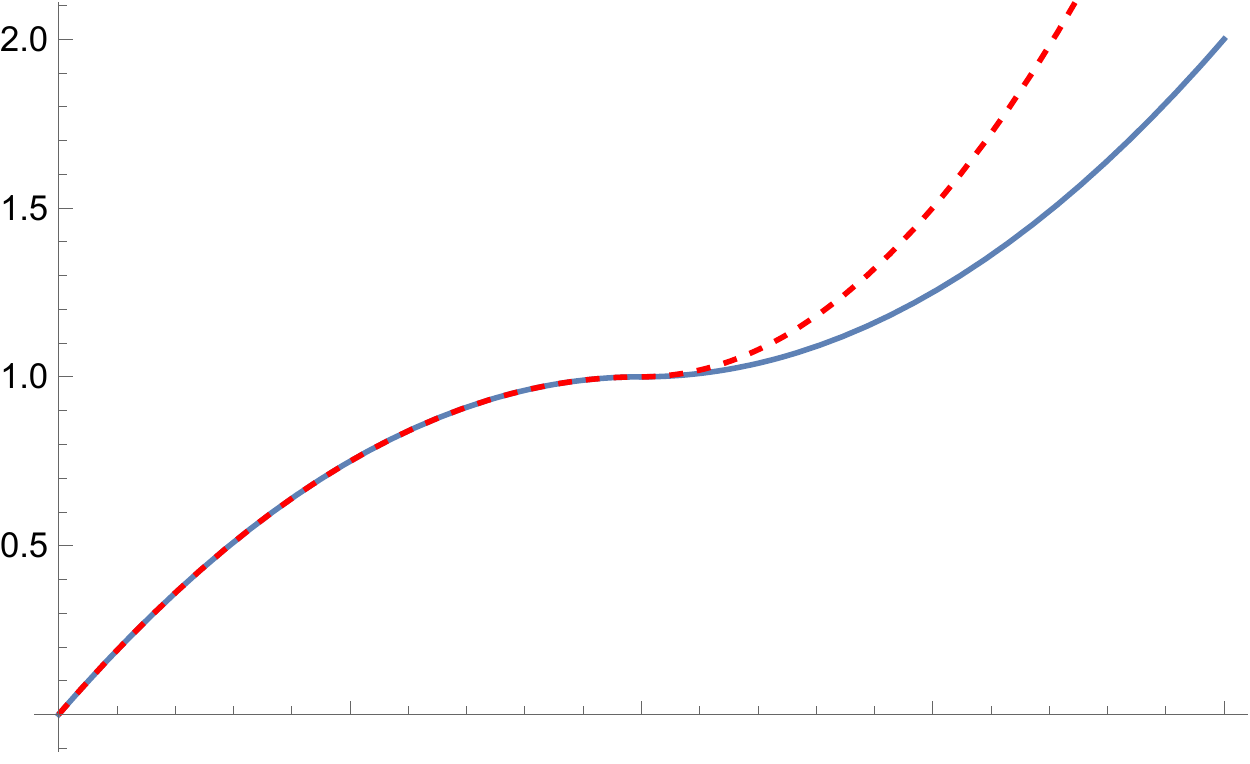}
\caption{
Plot of the cost functions in \eqref{eq:c1} (in blue) and \eqref{eq:c2} (in dotted red). 
The two functions are $\mathcal C^1$ but not convex.
}
\label{fig:non_differentiable_1}
\end{center}
\end{figure}

These functions are $\mathcal C^{1}$ and their derivatives vanish at $\load=1$. 
The equilibrium flows are:
\begin{center}

\def\arraystretch{2}
\begin{tabular}{c|c|c}
Interval &  $\load_{1}$ & $\load_{2}$\\
\hline
$\rate\in[0,2)$ & $\rate/2$ & $\rate/2$ \\
\hline
$\rate\in[2,+\infty)$ & $\frac{\sqrt{2}\cdot\rate-\sqrt{2}+1}{\sqrt{2}+1}$ & $\frac{\rate+\sqrt{2}-1}{\sqrt{2}+1}$ \\
\end{tabular}

\end{center}
Both flows are nondifferentiable at $\rate=2$. 
\end{example}

\begin{example}
\label{ex:non-differentiable-2}
Nonsmoothness can also be observed at a point where a flow vanishes. 
Consider the Wheatstone network in \cref{fig:braess-nondifferentiable} with
\begin{align*}
\cost_{1}(\load) &= \load,\\    
\cost_{2}(\load) &= \load,\\ 
\cost_{3}(\load) &= 
\begin{cases}
-\frac 1{10}(\load-1)^2+1 & \load\le 1,\\
10(\load-1)^2+1 & \load>1,
\end{cases}\\
\cost_{4}(\load) &= 
\begin{cases}
-(\load-1)^2+1 & \load\le 1,\\ 
(\load-1)^2+1 & \load>1,
\end{cases}\\
\cost_{5}(\load) &= \load^{2}.   
\end{align*}
\begin{figure}
\begin{center}
 \includegraphics[width=7cm]{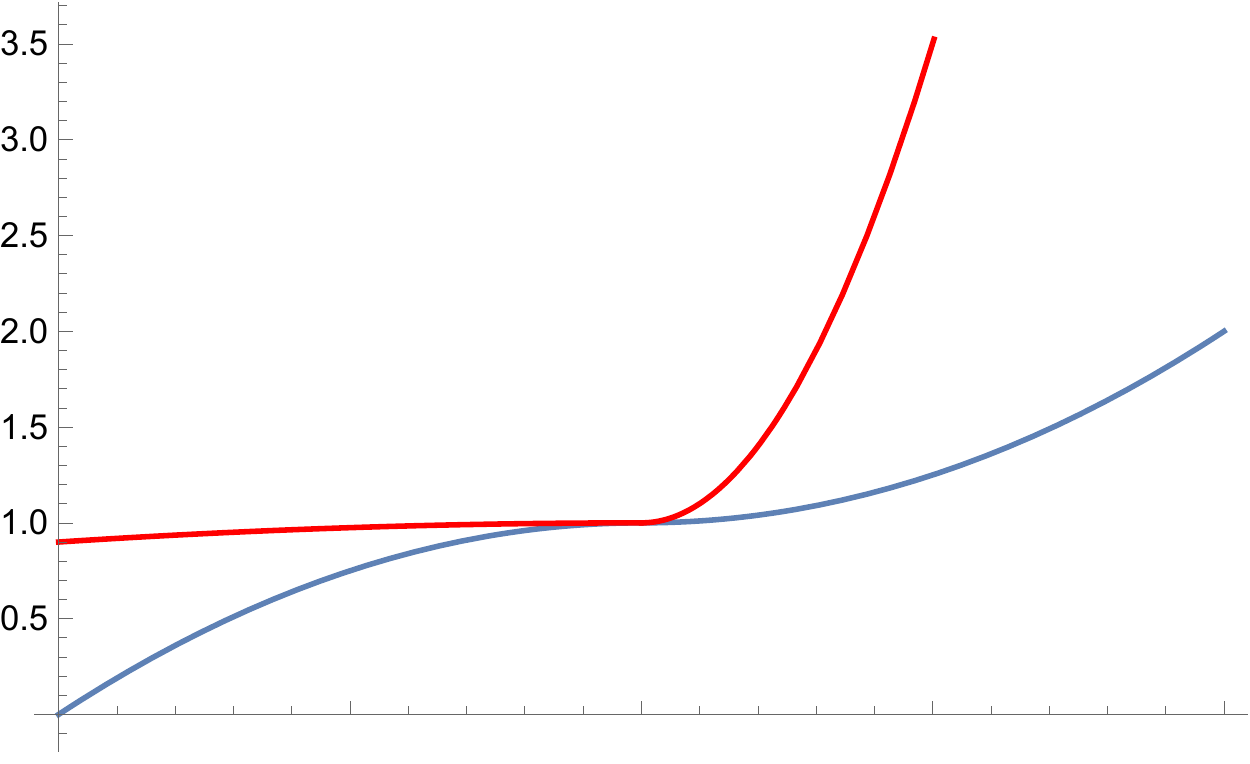}
\caption{
Plot of the cost functions $c_3$ (in red) and $c_4$ (in blue) in \cref{ex:non-differentiable-2}. 
The two functions are $\mathcal C^1$ but not convex.
}
\label{fig:non_differentiable-2}
\end{center}
\end{figure}
The plots of $\cost_3$ and $\cost_4$ can be found in \cref{fig:non_differentiable-2}.
All three paths are active in a neighborhood of $\rate=2$.
When the demand approaches $2$ from below, the flow on the zig-zag path decreases, and it vanishes when $\rate=2$, whereas it increases with positive derivative after $2$. 
A plot of the flow on the zig-zag path for $\rate\in[1,3]$ is shown in \cref{fig:braess-nondiff-flow}.

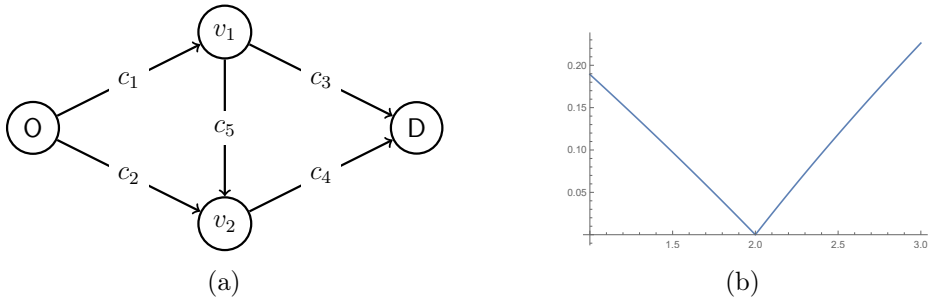
\begin{figure}[ht]
\setcounter{subfigure}{0}
\subfigure[]
{\scalebox{0.85}
{
\begin{tikzpicture}
   \node[shape=circle,draw=black,line width=1pt,minimum size=0.5cm] (v1) at (-3,0)  { $\source$}; 
   \node[shape=circle,draw=black,line width=1pt,minimum size=0.5cm] (v2) at (0,1.5)  {$\vertex_{1}$}; 

   \node[shape=circle,draw=black,line width=1pt,minimum size=0.5cm] (v5) at (0,-1.5)  {$\vertex_{2}$}; 
   \node[shape=circle,draw=black,line width=1pt,minimum size=0.5cm] (v6) at (3,0)  {$\sink$}; 
    
   \draw[line width=1pt,->] (v1) to   node[midway,fill=white] {$\cost_{1}$} (v2);
   \draw[line width=1pt,->] (v1) to   node[midway,fill=white] {$\cost_{2}$} (v5);
   \draw[line width=1pt,->] (v2) to   node[midway,fill=white] {$\cost_{5}$} (v5);

   \draw[line width=1pt,->] (v2) to 
   node[midway,fill=white] {$\cost_{3}$}(v6); 
   
	\draw[line width=1pt,->] (v5) to   node[midway,fill=white] {$\cost_{4}$} (v6);

\end{tikzpicture}
\label{fig:braess-nondifferentiable}
}}
\hspace{1cm}
\subfigure[]
{
\includegraphics[width=0.35\textwidth]{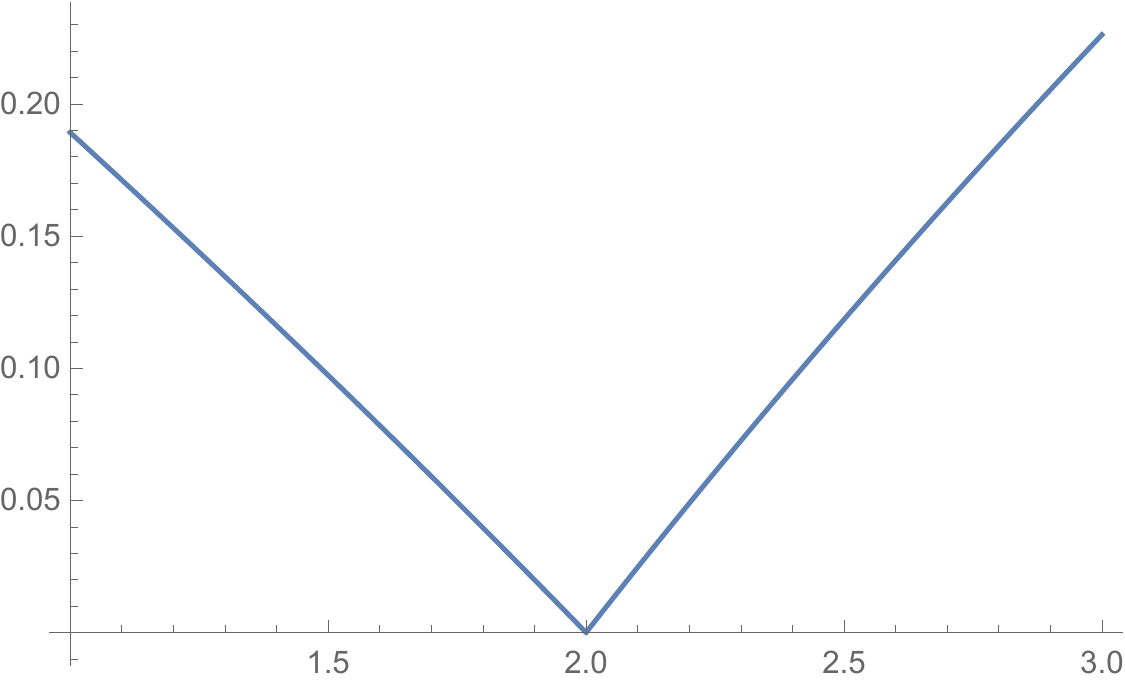}
\label{fig:braess-nondiff-flow}}
\caption{
On the left we have the classical Wheatstone network with the cost functions defined in \cref{ex:non-differentiable-2}. 
The plot on the right shows the load on the vertical link with demand varying in the interval $[1,3]$. Notice that at demand $2$ the load is zero and the load function is not differentiable.
}
\label{fig:Wheatstone}
\end{figure}
\end{example}

\section{Behavior around breakpoints}
\label{se:singleparameter}
In this section we describe the behavior of the functions $\eq{\SC}$ and $\PoA$ around breakpoints.
In particular we study their left and right derivatives 
when the demand  vector depends on a single real parameter. 
To formally describe  this  situation we make use of the following concept:

\begin{definition}
\label{def:breakpoints-on-curve}
We say that $\bar\var\in\Rpos$ is a \emph{$\activ{\routes}$-breakpoint} for $(\graph,\costprof,\ODpairs,\rateprof(\argdot))$ if there exist $\varepsilon>0$ such that $\activ{\routes}(\rateprof(\var))$ is constant as a function of $\var$ over each of the intervals $[\bar\var-\varepsilon,\bar\var)$ and $(\bar\var,\bar\var+\varepsilon]$ with $\activ{\routes}(\rateprof(\bar\var-\varepsilon))\ne\activ{\routes}(\rateprof(\bar\var+\varepsilon))$.
In this case, we use the symbols
\begin{equation}
\label{eq:left-right-regime} 
\activ{\routes}(\bar\var^{-})\coloneqq \activ{\routes}(\rateprof(\bar\var-\varepsilon)),
\quad
\activ{\routes}(\bar\var^{+})\coloneqq \activ{\routes}(\rateprof(\bar\var+\varepsilon)).
\end{equation}
\end{definition}

\citet{OHaConWat:TRB2016}, considered two regimes around a breakpoint and observed that, empirically, if one of these two regimes contains the other, then the derivatives of $\eq{\SC}$ and $\PoA$ associated to the smaller regime are larger than the derivatives in the larger regime.
They formalized this idea in their conjectures~4.5, 4.6, and 4.9. 
Using our language and notation, we can summarize these conjectures in the following statement, where we define  $g(\bar\var^{-}):=\lim_{\var\to\bar\var^{-}}g(\var)$ and $g(\bar\var^{+}):=\lim_{\var\to\bar\var^{+}}g(\var)$.

\begin{conjecture}[\citet{OHaConWat:TRB2016}]
\label{conj:Watling}
Let $(\graph,\costprof, \ODpairs, \rateprof(\argdot))$ be a routing game where the cost functions are continuous, differentiable, strictly increasing with positive second derivative, and the demand function $\var\rightarrow\rateprof(\var)$ is piecewise affine and componentwise nondecreasing with $\var$. If $\bar\var\in\Rpos$ is a $\activ{\routes}$-breakpoint, then the following hold:
\begin{enumerate}
[label={\textup{(\alph*)}}, ref={\textup{(\alph*)}}] 
\item
\label{eq:Watling-Inequality-1}
if $\activ{\routes}(\bar\var^{-})\subset\activ{\routes}(\bar\var^{+})$, then $\displaystyle(\eq{\SC}\circ\rateprof)'(\bar\var^{-})
>
(\eq{\SC}\circ\rateprof)'(\bar\var^{+})$ and
$(\PoA\circ\rateprof)'(\bar\var^{-})
>
(\PoA\circ\rateprof)'(\bar\var^{+})$;

\item
\label{eq:Watling-Inequality-2}
if $\activ{\routes}(\bar\var^{-})\supset\activ{\routes}(\bar\var^{+})$, then  $(\eq{\SC}\circ\rateprof)'({\bar\var}^{-})
<
(\eq{\SC}\circ\rateprof)'(\bar\var^{+})$ and
$(\PoA\circ\rateprof)'(\bar\var^{-})
<
(\PoA\circ\rateprof)'(\bar\var^{+})$,

\end{enumerate}

\end{conjecture}

As illustrated by the next example with affine costs and an affine demand function,
this conjecture does not hold in its most general form as stated above.
However, we will show later that a restricted form of the conjecture is indeed valid.

\begin{example}\label{ex:nonradial-smoothmaxPoA-counterWatling}

Consider the network in \cref{fig:modified-Fisk-1} studied in  \citet{Fis:TRB1979}, with affine cost functions and three \ac{OD} pairs whose demands and routes as shown in the next table:

\begin{center}
\begin{tabular}{c|c|l}
\ac{OD} &Demand& Feasible routes\\\hline
$(a,b)$ & $1$ & $\route_{1}=a\to b$\\
$(a,c)$ & $t$ & $\route_{2}=a\to c  \text{ and } \route_{3}=a\to b\to c$\\
$(b,c)$ & $100$ & $\route_4=b\to c$\\\hline
\end{tabular}
\end{center}

\vspace{1ex}
\noindent
where the map $\var\mapsto\rateprof(\var)$ is affine and only the demand of $(a,c)$ increases with $\var$. 

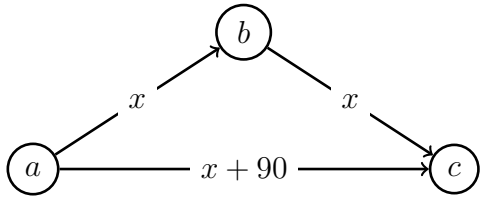
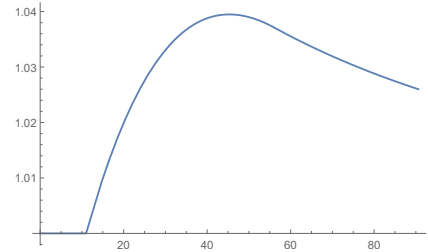
\begin{figure}[ht]
\setcounter{subfigure}{0}
\subfigure[Fisk's Network]
{
\begin{tikzpicture}[scale=0.7]
    \node[shape=circle,draw=black,line width=1pt] (v1) at (-4,0)  {$a$}; 
   \node[shape=circle,draw=black,line width=1pt] (v2) at (0,2.6)  { $b$}; 
   \node[shape=circle,draw=black,line width=1pt] (v6) at (4,0)  { $c$}; 
   \draw[line width=1pt,->] (v1) to   node[midway,fill=white] {$\load$} (v2);
   \draw[line width=1pt,->] (v1) to   node[midway,fill=white] {$\load+90$} (v6);
   \draw[line width=1pt,->] (v2) to   node[midway,fill=white] {$\load$} (v6);
\end{tikzpicture}
\label{fig:modified-Fisk-1}
}
\hspace{1cm}
\subfigure[\ac{PoA} for Fisk's Network with affine demand]
{
\includegraphics[width=0.4\textwidth]{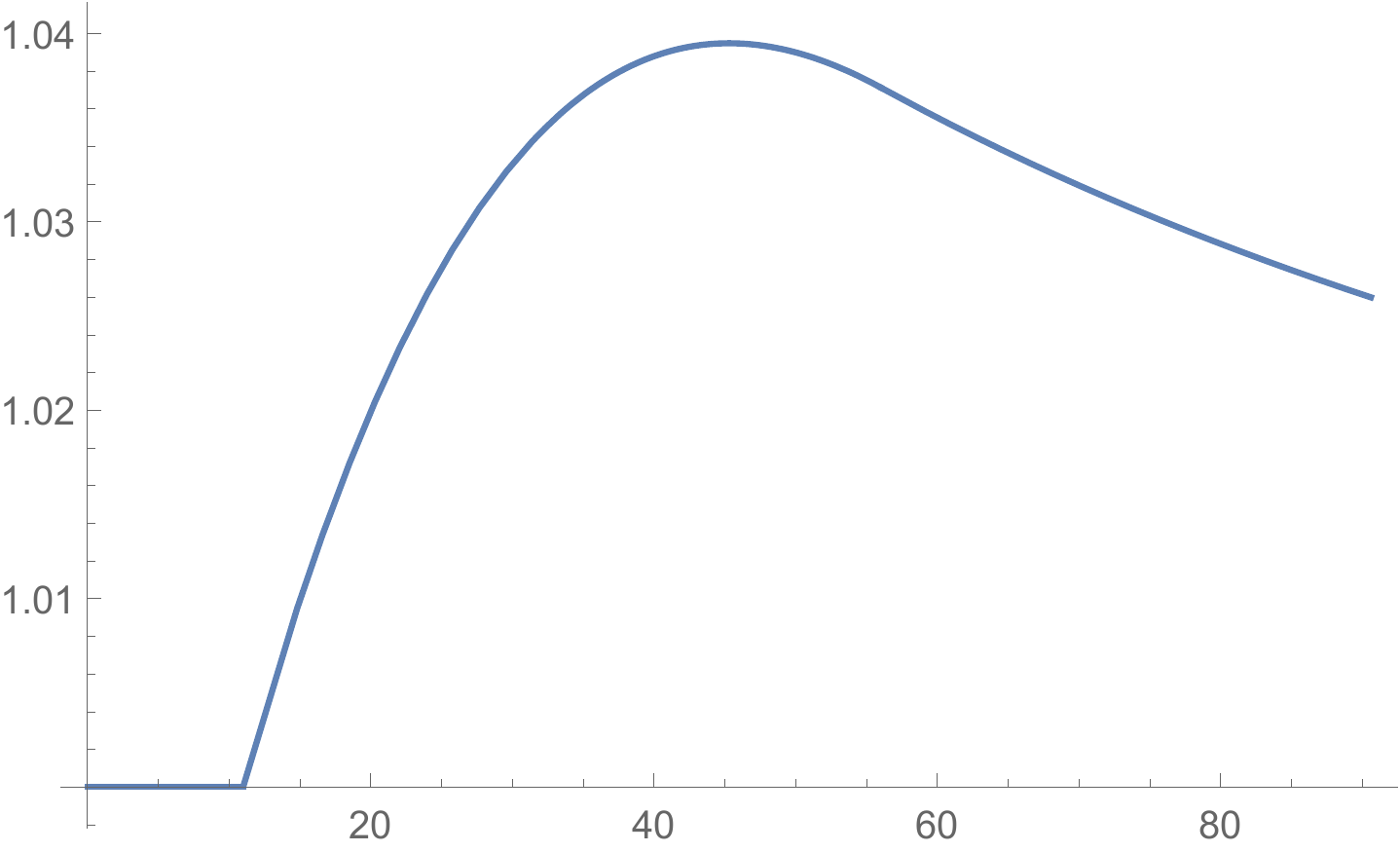}
\label{fig:modified-Fisk-2}}
\caption{An example where \cref{conj:Watling} does not hold. 
}
\label{fig:modified-Fisk}
\end{figure}

The following table shows the path flows and social cost at equilibrium 
 for  all $\var\geq 0$:

\begin{center}
\def\arraystretch{2}
\begin{tabular}{c|c|c|c|c|c}
Interval  
& $\eq{\flow}_{1}$ & $\eq{\flow}_{2}$ & $\eq{\flow}_{3}$ & $\eq{\flow}_4$ & $\eq{\SC}(\rate(\var))$ \\
\hline
$\var\in[0,11)$  
& 1 & $\var$ & 0 & 100 & $10001+90 \var+\var^2$\\
\hline
$\var\in[11,+\infty)$ 
& 1 & $\frac{1}3(11+2\var)$ & $\frac{1}3(\var-11)$ & 100 &$\frac{1}{3}(28892+382 \var+2 \var^2)$ 
\end{tabular}
\end{center}

\noindent whereas the optimum flow  and minimum social cost are given by

\begin{center}
\def\arraystretch{2}
\begin{tabular}{c|c|c|c|c|c}
Interval &  $\opt\flow_{1}$ & $\opt\flow_{2}$ & $\opt\flow_{3}$ & $\opt\flow_4$ & $\opt{\SC}(\rate(\var))$\\
\hline
$\var\in[0,56)$  & 1 & $\var$ & 0 & 100 &
$10001+90 \var+\var^2$ \\
\hline
$\var\in[56,+\infty)$ & 1 & $\frac{1}3(56+2\var)$ & $\frac{1}3(\var-56)$ & 100 &
$\frac{1}3(26867+382\var +2 \var^2)$.
\end{tabular}
\end{center}

\vspace{1ex}
\noindent
\cref{fig:modified-Fisk-2} shows the plot of the $\PoA$, which is
\begin{equation}\label{eq:PoA-modified-Fisk}
\PoA(\rate(\var))=
\begin{cases}
1 &\text{if }\var\in[0, 11),\\[3pt]
\dfrac{28892+382 \var+2 \var^2}{3(10001+90\var+\var^2)} &\text{if }\var\in[ 11,56),\\[10pt]

\dfrac{28892+382 \var+2 \var^2}{26867+382\var+2\var^2} &\text{if }\var\ge 56,
\end{cases}
\end{equation}
with derivative
\begin{equation}\label{eq:derivative-PoA-modified-Fisk}
(\PoA\circ\rate)'(\var)=
\begin{cases}
0 &\text{if }\var\in[0, 11),\\[3pt]
-\dfrac{2 (-610051 + 8890 \var + 101 \var^2)}{3 (10001 + 90 \var + \var^2)^2} &\text{if }\var\in[ 11,56),\\[10pt]

-\dfrac{4050 (191 + 2 \var)}{(26867 + 382 \var + 2 \var^2)^2} &\text{if }\var\ge 56.
\end{cases}
\end{equation}

At $\var=11$ the active regime experiences an expansion, where the route $\route_3$ becomes optimal for the pair $(a,c)$ at equilibrium. 
However, the right derivatives of $\eq{\SC}\circ\rateprof$ and $\PoA\circ\rateprof$ 
(respectively $142$ and $5/1852$)
are strictly larger than the respective left derivatives 
($112$ and $0$).
This provides a counter-example for \cref{conj:Watling}~\ref{eq:Watling-Inequality-1} as originally stated in \citet{OHaConWat:TRB2016}. Note also that, as predicted by \cref{lem:optimum-social-cost-smooth}, both $\opt{\SC}\circ\rateprof$ and $\PoA\circ\rateprof$ are smooth at $\var=56$, despite the presence of a phase transition in the social optimum.
For simplicity, in the example the demand was chosen  to move along the direction  $(0,1,0)$, however the counterexample also works for strictly positive directions  close to $(0,1,0)$.
\end{example}

We now proceed to prove that a restricted version of \cref{conj:Watling} is true, where the demand function is assumed to be linear, so that the proportion among the demands in each \ac{OD} pair is maintained constant as the total demand increases. 

\begin{theorem}
\label{th:WatlingSC}
Let $(\graph,\costprof, \ODpairs,\rateprof(\argdot))$ be a nonatomic routing game with a linear demand function
$\rateprof(\var)=\var\,\odrateprof$ for a fixed $\odrateprof\in\Rpos^{\ODpairs}$. Let $\bar\var\in\Rpos$ be a $\activ{\routes}$-breakpoint. 

\begin{enumerate}
[label={\textup{(\roman*)}}, ref={\textup{(\roman*)}}] 
\item
\label{it:th:WatlingSC-1}
If the edge cost functions are $\mathcal C^{1}$ with strictly positive derivatives,
 then properties \ref{eq:Watling-Inequality-1} and 
 \ref{eq:Watling-Inequality-2} in \cref{conj:Watling}  hold with weak inequalities.
    
\item
\label{it:th:WatlingSC-2}
If the edge costs are nondecreasing affine functions, then properties \ref{eq:Watling-Inequality-1} and 
 \ref{eq:Watling-Inequality-2} in  \cref{conj:Watling} hold with strict inequalities.
\end{enumerate}
\end{theorem}

\begin{remark}
\label{rm:WatlingWithoutCycles}
\cref{th:WatlingSC}\,\ref{it:th:WatlingSC-1}  can be generalized to networks with just strictly increasing cost functions, assuming that the set of edges $\edge$ where $\cost'_{\edge}(\load_{\edge}(\rateprof(\bar\var))=0$ form a graph without undirected cycles. 
\end{remark}

\begin{remark}
\label{re:BPR}
The hypothesis of strictly positive derivatives in \cref{th:WatlingSC}~\ref{it:th:WatlingSC-1} excludes from the result the case of \ac{BPR} cost functions. The result would hold also in the \ac{BPR} case if we would have a generalization of \cref{th:differentiability-on-curve}~\ref{th:differentiability-on-curve-a}  to strictly increasing smooth cost functions with zero derivative only possibly at zero. 
\end{remark}

\begin{remark}\label{rm:assumption-disjoint}
Violating the model assumption of having the sets $\routes^{(\odpair)}$ pairwise disjoint for every $\odpair\in\ODpairs$, does not affect the results of this paper. 
Indeed, one can always add a dummy origin $\tilde{\source}^{(\odpair)}$ and attach a zero cost edge $(\tilde{\source}^{(\odpair)},\source^{(\odpair)})$ to each of the original paths $\route\in\routes^{(\odpair)}$. 
This operation does not change the equilibrium;  hence the equilibrium loads on the original edges remain unchanged. 
Furthermore, this transformation does not create any new undirected cycle, thus, as noted in \cref{rm:WatlingWithoutCycles} the results are still valid.
\end{remark}

\begin{example}
\label{ex:contraction-and-expansion}
At a $\activ{\routes}$-breakpoint $\bar\var$ where neither $\activ{\routes}(\bar\var^{-})\subset\activ{\routes}(\bar\var^{+})$ nor $\activ{\routes}(\bar\var^{-})\supset\activ{\routes}(\bar\var^{+})$ are true, all inequalities are possible between the left and the right derivatives of the social cost and the \acl{PoA}. 
In order to illustrate this phenomenon, consider the  single \ac{OD} network in  \cref{fig:contr-and-exp} where one cost function depends on a parameter $\epsilon>0$. 
For every $\epsilon$, there is a breakpoint at $\bar\rate=2$ with neither $\activ{\routes}^{-}(\bar\rate)\subset\activ{\routes}^{+}(\bar\rate)$ nor $\activ{\routes}^{-}(\bar\rate)\supset\activ{\routes}^{+}(\bar\rate)$, and the left and right derivatives at $\bar\rate$ of $\eq{\SC}$ and $\PoA$ are ranked in different order depending on the value of $\epsilon$. 
There are four paths
\begin{align*}
\route_1=&\source\rightarrow\vertex_1\rightarrow\sink\\
\route_2=&\source\rightarrow\vertex_2\rightarrow\sink\\
\route_3=&\source\rightarrow\vertex_1\rightarrow\vertex_2\rightarrow\sink\\
\route_4=&\source\rightarrow\sink
\end{align*}
and the equilibrium flow vector $\flowprof$ is given in the following table:
\begin{center}
\def\arraystretch{2}
\begin{tabular}{c|c|c|c|c|c}
Interval & $\eqcost(\rate)$ & $\flow_{1}$ & $\flow_{2}$ & $\flow_{3}$ &$\flow_4$ 
\\
\hline
$\rate\in[0,\frac1{1+\varepsilon})$ & $(2+\varepsilon)\rate$ & 0 & 0 & $\rate$ &0
\\
\hline
$\rate\in[\frac1{1+\varepsilon},2)$ & $\frac{\varepsilon\rate+2+2\varepsilon}{1+2\varepsilon}$ & $\frac{(1+\varepsilon)\rate-1}{1+2\varepsilon}$ & $\frac{(1+\varepsilon)\rate-1}{1+2\varepsilon}$ & $\frac{2-\rate}{1+2\varepsilon}$&0    \\
\hline
$\rate\in[2,+\infty)$ & $\frac{\rate+4}3$ & $\frac{\rate+1}3$ & $\frac{\rate+1}3$ & 0 & $\frac{\rate-2}3$
\end{tabular}
\end{center}
where we can observe that 
\begin{equation*}
\lim_{\rate\to2^{-}}\eqcost'(\rate)=\frac{\varepsilon}{1+2\varepsilon},\quad
\lim_{\rate\to2^{+}}\eqcost'(\rate)=\frac13.  
\end{equation*}
Therefore,
\begin{align*}
\lim_{\rate\to2^{-}}\eqcost'(\rate)
&>
\lim_{\rate\to2^{+}}\eqcost'(\rate)
\quad\text{if }\varepsilon<1,\\
\lim_{\rate\to2^{-}}\eqcost'(\rate)
&=
\lim_{\rate\to2^{+}}\eqcost'(\rate)
\quad\text{if }\varepsilon=1,\\
\lim_{\rate\to2^{-}}\eqcost'(\rate)
&<
\lim_{\rate\to2^{+}}\eqcost'(\rate)
\quad\text{if }\varepsilon>1,
\end{align*}
which implies the same relations at $\bar\rate=2$ for the left and right derivatives  of $\eq{\SC}$ and $\PoA$.
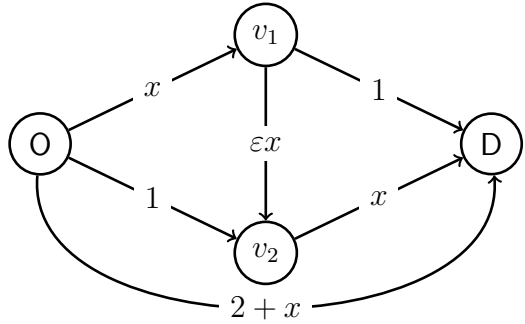
\begin{figure}

\begin{center}
\begin{tikzpicture}
   \node[shape=circle,draw=black,line width=1pt,minimum size=0.5cm] (v1) at (-3,0)  { $\source$}; 
   \node[shape=circle,draw=black,line width=1pt,minimum size=0.5cm] (v2) at (0,1.45)  {$\vertex_{1}$}; 

   \node[shape=circle,draw=black,line width=1pt,minimum size=0.5cm] (v5) at (0,-1.45)  {$\vertex_{2}$}; 
   \node[shape=circle,draw=black,line width=1pt,minimum size=0.5cm] (v6) at (3,0)  {$\sink$}; 
    
   \draw[line width=1pt,->] (v1) to   node[midway,fill=white] {$\load$} (v2);
   \draw[line width=1pt,->] (v1) to   node[midway,fill=white] {$1$} (v5);
   \draw[line width=1pt,->] (v2) to   node[midway,fill=white] {$\varepsilon\load$} (v5);
   
   \draw[line width=1pt,->] (v2) to   node[midway,fill=white] {$1$} (v6);

   \draw[line width=1pt,->] (v5) to   node[midway,fill=white] {$\load$} (v6);
   \draw[line width=1pt,->] (v1) to [bend right=95]  node[midway,fill=white] {$2+\load$} (v6);    
        
\end{tikzpicture}
\end{center}
\caption{In this network, varying the parameter $\epsilon$ we can observe any type of inequality between the left and right derivative of the equilibrium cost at a breakpoint}
\label{fig:contr-and-exp}
\end{figure}
\end{example}

As shown in \cref{th:WatlingSC}~\ref{it:th:WatlingSC-2}, with affine cost functions the inequalities of \cref{conj:Watling} can never hold with equality at a breakpoint. 
It is an open question whether this is also the case for convex costs or \ac{BPR} costs. 
On the other hand, the following example shows that the inequalities of \cref{conj:Watling} can indeed hold with equality if we allow for non-convex costs. 

\begin{example}\label{ex:Watling-equality}
Consider a single-commodity parallel network with two paths and cost functions:
\begin{equation}
\label{eq:c1-c2}
\cost_1(\load)=\begin{cases}
-(\load-1)^2+1&\text{if }\load\le 1,\\
(\load-1)^2+1&\text{if }\load> 1,
\end{cases}
\quad
\cost_2(\load)=1+x^2.    
\end{equation}
For $\rate\le 1$ only the first link is used, and the equilibrium cost $\eqcost(\rate)$ is equal to $-(\rate-1)^2+1$, whose derivative at $1$ is zero. 
For $\rate\in[1,+\infty)$ the equilibrium flow routes $(\rate+1)/2$ through the first link and $(\rate-1)/2$ through the second link. 
Hence, the equilibrium cost is
\begin{equation*}
\eqcost(\rate)=1+\parens*{\frac{\rate-1}2}^2    
\end{equation*} 
for every $\rate\in[1,+\infty)$, whose derivative at $1$ also equals  zero.
\end{example}

\section{Proof of \cref{th:WatlingSC}}
\label{se:proof-main-theorem}
The proof of \cref{th:WatlingSC} will be split into a number of preliminary intermediate steps.
The general principle is described in \cref{le:GeneralWatlingLemma} which assumes \emph{a priori} the existence of a selection of equilibrium flows that is smooth both to the left and to the right of a breakpoint. 
This assumption is then shown to hold both when the  cost functions have strictly positive derivatives (\cref{th:WatlingSC}~\ref{it:th:WatlingSC-1}), and also when the costs are affine (\cref{th:WatlingSC}~\ref{it:th:WatlingSC-2}).

\begin{lemma}
\label{le:GeneralWatlingLemma}
Let $(\graph,\costprof,\ODpairs,\rateprof(\argdot))$ be a nonatomic routing game with $\mathcal C^{1}$ nondecreasing costs and a linear demand function
$\rateprof(\var)=\var\,\odrateprof$
 with $\odrateprof\in\Rpos^{\ODpairs}$. 
Let $\bar\var\in\Rpos$ be a $\activ{\routes}$-breakpoint with active regimes 
$\regime^{-} \coloneqq \activ{\routes}(\bar\var^{-})$ over $\inter^{-} \coloneqq (\bar\var-\varepsilon,\bar\var)$   and $\regime^{+} \coloneqq \activ{\routes}(\bar\var^{+})$  over $\inter^{+} \coloneqq (\bar\var, \bar\var+\varepsilon)$.
Suppose that there exist  maps
$\flowprof^{-}, \flowprof^{+}: (\bar\var-\varepsilon,\bar\var+\varepsilon) \to \Rpos^{\routes}$ 
 of class $\mathcal{C}^1$ such that
\begin{itemize}
\item 
$\flowprof^{-}(\var)$ is an equilibrium flow of demand $\rateprof(\var)$ for all $\var\in \inter^{-}$,

\item 
 $\flowprof^{+}(\var)$ is an equilibrium flow of demand $\rateprof(\var)$ for all $\var\in \inter^{+}$,
 \item  $\flowprof^{-}(\bar\var)=\flowprof^{+}(\bar\var)$.
\end{itemize}
Then \cref{conj:Watling} holds with weak inequalities.
\end{lemma}

\begin{proof}
Let $\loadprof^{-}(\var)$ and $\loadprof^{+}(\var)$ be the loads induced by $\flowprof^{-}(\var)$ and $\flowprof^{+}(\var)$, and
let $\bar\loadprof$ and $\bar\flowprof$ be their corresponding common values at $\bar\var$.
For notational simplicity, for each $\var\in(\bar\var-\varepsilon,\bar\var+\varepsilon)$ 
we let 
$\eqcost^{(\odpair)}(\var)=\eqcost^{(\odpair)}(\rateprof(\var))$ denote the equilibrium costs and $\eq{\SC}(\var)=\eq{\SC}(\rateprof(\var))$ the social cost at equilibrium with demand $\rateprof(\var)$,
so that
\begin{equation}
\label{eq:SC-mu}
\eq{\SC}(\var)=\sum_{\odpair\in\ODpairs}\rate^{(\odpair)}(\var)\cdot\eqcost^{(\odpair)}(\var).    
\end{equation}
By \cref{lem:continuity}, the maps $\var\mapsto\eqcost^{(\odpair)}(\var)$ are continuous. 
They are also $\mathcal{C}^{1}$ over both $\inter^{-}$ and $\inter^{+}$, with 
well defined unilateral limits $(\eqcost^{(\odpair)})'(\bar\var^{-})$ and $(\eqcost^{(\odpair)})'(\bar\var^{+})$ at $\bar\var$.
Indeed, for each $\odpair\in\ODpairs$ we may fix an active path $\route\in\regime^{+}$ so that, for all $\var\in \inter^{+}$, we have 
$\eqcost^{(\odpair)}(\var)=\cost_{\route}(\loadprof^{+}(\var))$, 
which is smooth with a well defined limit for its derivative at $\bar\var^{+}$. A similar argument using $\loadprof^{-}(\argdot)$ applies for $\var\in \inter^{-}$. 

Since $\rate^{(\odpair)}(\var)=\var\,\odrate^{(\odpair)}$, we can express the left and the right derivatives at $\bar\var$ of the social cost as
\begin{align*}
(\eq{\SC})'(\var^{-})=&
\sum_{\odpair\in\ODpairs}\odrate^{(\odpair)}\cdot\eqcost^{(\odpair)}(\bar\var)+
\bar\var\,\odrate^{(\odpair)}\cdot(\eqcost^{(\odpair)})'(\bar\var^{-}),
\\
(\eq{\SC})'(\var^{+})=&
\sum_{\odpair\in\ODpairs}\odrate^{(\odpair)}\cdot\eqcost^{(\odpair)}(\bar\var)+
\bar\var\,\odrate^{(\odpair)}\cdot(\eqcost^{(\odpair)})'(\bar\var^{+}).
\end{align*}
It follows that $(\eq{\SC})'(\var^{-})\ge (\eq{\SC})'(\var^{+})$ if and only if
$\thfunc(\bar\var^{-})
\geq\thfunc(\bar\var^{+})$, where for $\var\neq\bar\var$ we define
\begin{equation}
\label{eq:des}
\thfunc(\var) \coloneqq \sum_{\odpair\in\ODpairs}
\odrate^{(\odpair)} \cdot(\eqcost^{(\odpair)})'(\var).
\end{equation}

We now  derive an alternative characterization for $\thfunc(\bar\var^{-})$ and $\thfunc(\bar\var^{+})$.
Let us consider the latter.
For  $\var\in \inter^{+}$ we have that
$\flowprof(\var)\coloneqq\flowprof^{+}(\var)$ is an equilibrium flow with associated induced equilibrium load $\loadprof(\var)\coloneqq\loadprof^{+}(\var)$. In particular, 
\begin{equation}
\label{eq:cond-equil-I}
\sum_{\route\in\routes^{(\odpair)}}\flowprof_{\route}(\var)=\rate^{(\odpair)}(\var)=\var\,\odrate^{(\odpair)},    
\end{equation}
so that differentiating we get
$\sum_{\route\in\routes^{(\odpair)}}\flowprof_{\route}'(\var)=\odrate^{(\odpair)}.  
$
Therefore for $\var>\bar\var$ we have
\begin{equation}
\label{eq:theta}
\thfunc(\var)=\sum_{\odpair\in\ODpairs}
\sum_{\route\in\routes^{(\odpair)}}\flowprof_{\route}'(\var)\cdot(\eqcost^{(\odpair)})'(\var).
\end{equation}
Now, for any inactive path $\route\not\in\regime^{+}$ we have  $\flow_{\route}'(\var)=0$,
so the inner sum 
in \eqref{eq:theta} can be restricted to the active paths $\route\in\regime^{+}$. For these  paths $\route\in\routes^{(\odpair)}\cap\regime^+$ we have 
$\eqcost^{(\odpair)}(\var)=\cost_{\route}(\loadprof(\var))$ and  therefore
\begin{equation}
\label{eq:lambda-mu}
(\eqcost^{(\odpair)})'(\var)
=\sum_{\edge\in\route}\cost'_{\edge}(\load_\edge(\var))\cdot\load_{\edge}'(\var)
=\sum_{\edge\in\route}\cost'_{\edge}(\load_\edge(\var))\cdot\sum_{\routealt\in\routes}\delta_{\edge\routealt}\flow_{\routealt}'(\var).
\end{equation}
Plugging this into \eqref{eq:theta} and using Fubini's rule to exchange the order of sums, we
can write $\thfunc(\var)$ in the
equivalent form
\begin{equation}
\label{eq:theta_expression_general}
\begin{split}
\thfunc(\var)
&=\sum_{\edge\in\edges}\cost_{\edge}'(\load_{\edge}(\var))
\parens*{\sum_{\route\in\routes}\edgepath_{\edge\route}\flow_{\route}'(\var)}^2.
\end{split}
\end{equation}

Let us recall that $\bar\loadprof$ denotes the equilibrium load induced by $\bar\flowprof=\flowprof(\bar\var)$, and consider the positive semidefinite quadratic form

\begin{equation}
\label{eq:quad_opt}
Q_{\bar\loadprof}(\yvec)\coloneqq\frac{1}{2}\sum_{\edge\in\edges}\cost_{\edge}'(\bar\load_{\edge})
\parens*{\sum_{\route\in\routes}\edgepath_{\edge\route}\yvar_{\route}}^2\geq 0.
\end{equation}
We claim that $\thfunc(\bar\var^{+})$ coincides with the optimal value of the convex quadratic program
\begin{equation}
\label{eq:program_fprime_+}
\thfunc(\bar\var^{+})=\min_{\yvec}
~~Q_{\bar\loadprof}(\yvec)
\quad\text{s.t.}\
\sum_{\route\in\routes^{(\odpair)}}\yvar_{\route}=\odrate^{(\odpair)}
\ \text{for all }\odpair\in\ODpairs,
\ \text{and }\yvar_{\route}=0
\ \text{for }\route\notin\regime^{+},
\end{equation}
with optimal solution $\bar{\yvec}=\flowprof'(\bar\var)$.
Indeed,  we already observed that $\flowprof'(\var)$ satisfies the constraints in \eqref{eq:program_fprime_+} for all $\var>\bar\var$, so that letting $\var\downarrow\bar\var$ the same holds for $\bar\yvar=\flowprof'(\bar\var)$. 
On the other hand, denoting $\zvar_\edge=\sum_{\route\in\routes}\delta_{\edge\route}\yvar_{\route}$, we have ${\partial Q_{\bar\loadprof}}/{\partial\yvar_{\route}}=\sum_{\edge\in\route}\cost_{\edge}'(\bar\loadprof)\,\zvar_\edge$, so that the first-order optimality conditions for \eqref{eq:program_fprime_+} are  
\begin{equation}
\label{eq:condition_fprime}
\forall \odpair\in\ODpairs, \forall\route\in\routes^{(\odpair)},
\quad\sum_{\edge\in\route}\cost_{\edge}'(\bar\loadprof)\,\zvar_\edge
=
\begin{cases}
\multod^{(\odpair)} &\text{if }  \route\in\regime^{+},\\
\multod^{(\odpair)} +\multpath_{\route}&\text{if }  \route\not\in\regime^{+},
\end{cases}
\end{equation}
where the $\multod^{(\odpair)}$ are Lagrange multipliers for the constraints $\sum_{\route\in\routes^{(\odpair)}}\yvar_{\route}=\odrate^{(\odpair)}$ for all $\odpair\in\ODpairs$,  and $\multpath_{\route}$ are multipliers for the constraints $\yvar_{\route}=0\text{ for }\route\notin\regime^{+}$.
Now, since  for all $\var\in \inter^{+}$ the pair $(\flowprof(\var),\loadprof(\var))$ is an equilibrium flow with $\activ{\routes}(\rateprof(\var))=\regime^{+}$, we have
\begin{equation}
\label{eq:cf_cont}
\forall \odpair\in\ODpairs,\forall\route\in\routes^{(\odpair)}
\quad\sum_{\edge\in\route}\cost_{\edge}(\load_\edge(\var))
=
\begin{cases}
\eqcost^{(\odpair)}(\var) &\text{if }  \route\in\regime^{+},\\
\eqcost^{(\odpair)}(\var) +(\cost_{\route}(\loadprof(\var)) -\eqcost^{(\odpair)}(\var))&\text{if }  \route\not\in\regime^{+}.
\end{cases}
\end{equation}
Differentiating in $\var$ and letting $\var\downarrow\bar\var$, we deduce that
\eqref{eq:condition_fprime} holds with 
$\multod^{(\odpair)}=(\eqcost^{(\odpair)})'(\bar\var^{+})$ for all $\odpair\in\ODpairs$
and $\multpath_{\route}=(\cost_{\route}\circ\loadprof)'(\bar\var)-(\eqcost^{(\odpair)})'(\bar\var^{+})$ for $\route\in\routes^{(\odpair)}\setminus\regime^{+}$.
Since the quadratic program is convex, these optimality conditions imply optimality. Hence $\bar\yvec=\flowprof'(\bar\var)$ is an optimal solution and then the equality \eqref{eq:program_fprime_+} follows 
 by letting $\var\to\bar\var^{+}$  in \eqref{eq:theta_expression_general}.

Repeating the argument for $\var\in \inter^{-}$, this time with $\flowprof(\var)=\flowprof^{-}(\var)$ and $\loadprof(\var)=\loadprof^{-}(\var)$, and noting that $\loadprof^-(\bar\var)=\bar\loadprof$,
we obtain that $\thfunc(\bar\var^{-})$ is similarly characterized as 
\begin{equation}
\label{eq:program_fprime_-}
\thfunc(\bar\var^{-})=\min_{\yvec} ~~Q_{\bar\loadprof}(\yvec)
\quad\text{s.t.}\
\sum_{\route\in\routes^{(\odpair)}}\yvar_{\route}=\odrate^{(\odpair)}
\ \text{for all }\odpair\in\ODpairs,
\ \text{and }\yvar_{\route}=0
\ \text{for }\route\notin\regime^{-}.
\end{equation}

The objective function in both \eqref{eq:program_fprime_+} and \eqref{eq:program_fprime_-} is the same, and the only difference between 
these quadratic programs are the constraints $\yvar_{\route}=0$ for all $\route\not\in\regime^{-}$ or $\route\not\in\regime^{+}$: 
a larger regime implies fewer constraints  and thus a smaller optimal value. 
Explicitly, if $\regime^{-}\subset \regime^{+}$, then $\thfunc(\bar\var^{-})\geq \thfunc(\bar\var^{+})$; and therefore
$(\eq{\SC})'(\bar\var^{-})
\geq
(\eq{\SC})'(\bar\var^{+})$.
Similarly, if $\regime^{-}\supset \regime^{+}$, then $\thfunc(\bar\var^{-})\leq \thfunc(\bar\var^{+})$ and 
$(\eq{\SC})'(\bar\var^{-})
\leq
(\eq{\SC})'(\bar\var^{+})$.
This establishes half of \cref{conj:Watling} with weak inequalities.

Concerning the derivative of $\PoA\circ\rateprof$, when it exists, it is
\begin{equation}
\label{eq:deriv-PoA-mu}
(\PoA\circ\rateprof)'(\var)=\frac{(\eq{\SC}\circ\rateprof)'(\var)\cdot \opt{\SC}(\rateprof(\var))-\eq{\SC}(\rateprof(\var))\cdot(\opt{\SC}\circ\rateprof)'(\var)}{(\opt{\SC}(\rateprof(\var)))^2}.  
\end{equation}
Using the continuity of $\eq{\SC}$ at $\rateprof(\bar\var)$ 
and the fact that $\opt{\SC}$ is $\mathcal C^{1}$ at $\rateprof(\bar\var)$ (\cref{lem:continuity,lem:optimum-social-cost-smooth}), we  see that the statements regarding the \acl{PoA} in \cref{conj:Watling} also hold with weak inequalities. 
\end{proof}

In order to use \cref{le:GeneralWatlingLemma} one needs to ensure the existence of smooth equilibrium flows to the left and to the right of a breakpoint. Having a smooth equilibrium load profile (which can be derived from the results in \cref{se:regularity}) allows us to choose a smooth equilibrium flow, as shown in the following lemma.
Note that even when the equilibrium loads are unique, there may exist multiple path flows that induce the same loads. 
In order to single out a smooth selection of path flows we  use the Moore-Penrose pseudoinverse. 
We recall that the pseudoinverse
generalizes the classical inverse of nonsingular square matrices, and provide a tool for solving an ill-posed linear system $A\,\boldsymbol{x}=\boldsymbol{b}$, where $A$ is an $m\times n$ rectangular real matrix and the right-hand side $\boldsymbol{b}\in\reals^{m}$ may not even belong to the range of $A$ so that the system could have no solution. 
The pseudo-inverse is a linear map $\boldsymbol{b}\mapsto A^{\dag}\boldsymbol{b}$ that associates to each $\boldsymbol{b}\in\reals^m$ the least norm solution of the equation $A\,\boldsymbol{x}=\bar{\boldsymbol{b}}$ with $\bar{\boldsymbol{b}}$ the projection of $b$ onto  $\Range(A)$. 
Thus, if $\boldsymbol{b}\in \Range(A)$ it gives the least norm solution of $A\,\boldsymbol{x}=\boldsymbol{b}$ and when $A$ is an invertible square matrix it reduces to the classical inverse $A^{\dag}\boldsymbol{b}=A^{-1}\boldsymbol{b}$.
In general, $A^{\dag}\boldsymbol{b}$ is defined as the limit when $\varepsilon\to 0^+$ of the unique minimizer of the strongly convex function $\boldsymbol{x}\mapsto \|A\boldsymbol{x}-\boldsymbol{b}\|^2+\varepsilon\|\boldsymbol{x}\|^{2}$, that is, $A^\dag \boldsymbol{b}=\lim_{\varepsilon\to 0^+}(A^{\top}A+\varepsilon I)^{-1}A^{\top}\boldsymbol{b}$. 
This defines a linear map whose representative matrix is denoted by $A^{\dag}$. 
In our setting we  use the pseudoinverse to associate with each vector of loads $\loadprof$ the vector of path flows $\flowprof$
that induce the given $\loadprof$ and has minimum norm $\|\flowprof\|$. 

\begin{lemma}
\label{lem:pseudoinverse}
Let $(\graph,\costprof,\ODpairs,\rateprof(\argdot))$ be a nonatomic routing game with a $\mathcal{C}^1$ demand map $\rateprof(\argdot)$. Let $\var\mapsto\loadprof(\rateprof(\var))$ be a curve  of equilibrium loads which is defined and $\mathcal{C}^1$ on some interval $I$. 
Then, $\loadprof(\rateprof(\var))$ can be decomposed into a path flow $\flowprof(\rateprof(\var))$ that is also $\mathcal{C}^1$ for $\var\in I$.
\end{lemma}

\begin{proof}
Let $\flowprof^0\in\mathbb R^\npaths$ be an arbitrary path flow decomposition of $\load(\rateprof(\var^0))$
with $\var^0\in I$ fixed, and set
\begin{equation}
\label{eq:f-mu}
\flowprof(\rateprof(\var)) \coloneqq \flowprof^0+
\begin{bmatrix}
\edgepaths\\
\sumflow
\end{bmatrix}^{\dag}
\begin{pmatrix}
\loadprof(\rateprof(\var))-\loadprof(\rateprof(\var^0))\\ 
\rateprof(\var)-\rateprof(\var^0)
\end{pmatrix},
\end{equation}
where $\edgepaths$ and $\sumflow$ are as in \eqref{eq:matricesConstraints},  and the symbol $L^{\dag}$ denotes the Moore-Penrose pseudoinverse
of the matrix $L$. 
Since the pseudoinverse is also a matrix, we get the conclusion.
\end{proof}

With these preliminaries, we may now proceed with the proof of \cref{th:WatlingSC}.

\begin{proof}[Proof of \cref{th:WatlingSC}~\ref{it:th:WatlingSC-1}]
Using \cref{lem:differentiability-of-fixed-regime-possibly-negative-flows}
we have that
for every regime $\regime$ one can find  a $\mathcal C^{1}$ solution $\var\mapsto\loadprof_{\regime}(\rateprof(\var))$ to the problem \eqref{eq:Beckmann-fixed-regime-relaxed}, 
defined in an open neighborhood of $\bar\var$, and such that $\loadprof_{\regime}(\rateprof(\var))$ coincides with the equilibrium $\loadprof(\rateprof(\var))$ whenever $\activ{\routes}(\rateprof(\var))$ is equal to $\regime$. 
Furthermore,  by applying \cref{lem:pseudoinverse} we can find $\mathcal C^{1}$ flows associated to such load profiles.
If we do that in the two cases of $\regime=\activ{\routes}(\bar\var^{-})$ and $\regime=\activ{\routes}(\bar\var^{+})$, we obtain differentiable flows in a neighborhood of $\bar\var$ that are equilibria wherever the active regime coincides with  $\regime$. 
Now, because the edge costs are strictly increasing, the equilibrium load vector $\bar{\loadprof}$ is unique, and then it follows by continuity that the flows satisfy
$\flowprof^{-}(\rateprof(\bar\var))=\flowprof^{+}(\rateprof(\bar\var))$, so that 
we can apply \cref{le:GeneralWatlingLemma} to conclude the proof.
\end{proof}

\begin{proof}[Proof of  \cref{th:WatlingSC}~\ref{it:th:WatlingSC-2}]
When the cost functions are affine, to the left and to the right of $\bar\var$ we can choose equilibrium flows in affine form
\begin{equation}
\label{eq:equil-affine-form}
\flowprof^{-}(\rateprof(\var))=\vecw^{-}\cdot\var+\vecz^{-},\quad
\flowprof^{+}(\rateprof(\var))=\vecw^{+}\cdot\var+\vecz^{+},
\end{equation}
with $\flowprof^{-}(\rateprof(\bar\var))=\flowprof^{+}(\rateprof(\bar\var))$. A proof of this in the case of a single commodity can be found in \citet[proposition~4.1]{ComDosScar:MathPRog2021}. The reasoning for the multi-commodity case is completely analogous.

These functions are  differentiable as functions of $\var$ so that \cref{le:GeneralWatlingLemma} implies that \ref{eq:Watling-Inequality-1} and \ref{eq:Watling-Inequality-2} in \cref{conj:Watling} hold with weak inequalities.
It remains to show that the inequalities are strict.
By arguing as in the proof of \cref{le:GeneralWatlingLemma}, we need to show that the inequality between $\thfunc(\bar\var^{-})$ and $\thfunc(\bar\var^{+})$ is strict, 
which requires to compare the values of the problems \eqref{eq:program_fprime_+} and \eqref{eq:program_fprime_-}. 
As one can deduce looking at \eqref{eq:quad_opt}, the two problems are exactly the problems $\parens*{\prim_{\rateprof(\bar\var)}^{\regime^+}}$ and $\parens*{\prim_{\rateprof(\bar\var)}^{\regime^-}}$ as in \cref{lem:differentiability-of-fixed-regime-possibly-negative-flows} for a game on the graph $\graph$, where the cost functions are the linear functions $\bar\cost_{\edge}(\load)=\cost'_{\edge}(\bar\load_{\edge})\cdot \load$, with $\bar\load_{\edge}=\sum_{\route\ni\edge}\flow^-_{\route}(\rate(\bar\var))=\sum_{\route\ni\edge}\flow^+_{\route}(\rate(\bar\var))$.

In these two problems, the objective function is the same, but there are more constraints in the case associated to the smaller regime between $\regime^{-}$ and $\regime^{+}$.
Suppose for instance that $\regime^{-}$ is strictly contained in $\regime^{+}$. 
Then the extra constraints in  \eqref{eq:program_fprime_-} with respect to \eqref{eq:program_fprime_+} are $\yvar_{\route}=0$ for $\route\in\regime^{+}\setminus\regime^{-}$,
with associated Lagrange multipliers $\multpath_{\route}$ given as in the proof of \cref{le:GeneralWatlingLemma}, that is, denoting $\odpair(\route)$ the commodity associated with the route $\route$,
\begin{equation}\label{eq:multiplier}
\multpath_{\route}=(\costprof_{\route}\circ\loadprof^{+})'(\bar\var)-(\eqcost^{(\odpair(\route))})'(\bar\var^{+}).
\end{equation}
The path $\route$ is  active when $\var>\bar\var$, but not when $\var<\bar\var$. 
Hence, there exists $\varepsilon>0$ such that
\begin{align}
\label{eq:cp-f-ineq}
(\cost_{\route}\circ\loadprof^{+})(\var)
&>\eqcost^{(\odpair(\route))}(\var), \quad\text{for } \var \in (\bar\var-\varepsilon,\bar\var) 
\intertext{because $\route$ is not active, and}
\label{eq:cp-f-eq}
(\cost_{\route}\circ\loadprof^{+})(\bar\var)
&=\eqcost^{(\odpair(\route))}(\bar\var)
\end{align}
by continuity of path costs. 

By assumption, the cost functions are affine; hence the path costs and equilibrium costs are also affine in the variable $\var\in (\bar\var-\varepsilon,\bar\var)$. 
This implies that the difference  $(\cost_{\route}\circ\loadprof^{+})'(\var)-(\eqcost^{(\odpair(\route))})'(\var)$ is a strictly negative constant on $(\bar\var-\varepsilon,\bar\var)$. 
Hence, the Lagrange multiplier $\multpath_{\route}$ assumes a strictly negative value at the optimum solution. 
Since the multipliers are unique for the problem $\parens*{\prim_{\rateprof(\bar\var)}^{\regime^-}}$ by \cref{lem:differentiability-of-fixed-regime-possibly-negative-flows}, relaxing the constraint $\yvar_{\route}=0$  in \eqref{eq:program_fprime_-} will produce a strict decrease in the optimum value of the objective function, that is, $\thfunc(\bar\var^{-})>\thfunc(\bar\var^{+})$,
which gives in turn the strict inequalities in the statement \ref{eq:Watling-Inequality-1} of \cref{conj:Watling}.

A similar argument holds 
for case
\ref{eq:Watling-Inequality-2} in \cref{conj:Watling}
when 
$\regime^{-}$ strictly contains $\regime^{+}$.
\end{proof}

\begin{remark}
To have statements \ref{eq:Watling-Inequality-1} and \ref{eq:Watling-Inequality-2} in \cref{conj:Watling}, linearity of the demand function is necessary. Indeed, as  shown by \cref{ex:nonradial-smoothmaxPoA-counterWatling},  \cref{le:GeneralWatlingLemma} can fail
even if the demand function is affine
(and even when restricting to affine cost functions).
Hence  \cref{th:WatlingSC} cannot be applied.
\end{remark}

\subsection*{Acknowledgments}
We thank two reviewers for their careful reading of the paper and their useful remarks and suggestions.   
Valerio Dose and Marco Scarsini are members of GNAMPA-INdAM. 
Their work was partially supported by the GNAMPA project CUP\_E53C22001930001 \emph{``Limiting behavior of stochastic dynamics in the Schelling segregation model.''} 
Marco Scarsini's work is partially supported by the MIUR PRIN 2022EKNE5K grant \emph{``Learning in markets and society.''}
Roberto Cominetti' research was supported by Proyecto Anillo ANID/PIA/ACT192094.

% Appendix
\appendix

\gdef\thesection{\Alph{section}} % corrected redefinition of "\thesection"
\makeatletter
\renewcommand\@seccntformat[1]{\appendixname\ \csname the#1\endcsname.\hspace{0.5em}}
\makeatother

\section{Missing proofs}
\label{se:supplementary}

\begin{proof}[Proof of \cref{lem:optimum-social-cost-smooth}]
The assumptions on the cost functions imply that the optimum flows are the Wardrop equilibria of the game $(\graph,\opt\costprof,\ODpairs)$, where
\begin{equation}
\label{eq:marginal-game}    
\opt\cost_{\edge}(\load_{\edge})\coloneqq \cost_{\edge}(\load_{\edge})+\load_{\edge}\cost'_{\edge}(\load_{\edge}).
\end{equation}
For every $\edge\in\edges$, the function $\opt\cost_{\edge}$ is continuous and nondecreasing. 
For this modified game, the minimal value of the Beckmann potential is 
\begin{equation}
\label{eq:min-Beckmann-marg}    
\opt \valueV(\rateprof)=\min_{\loadprof\in\loads_\rate}\sum_{\edge\in\edges}\load_{\edge}\cost_{\edge}(\load_{\edge}),
\end{equation}
which is the optimum social cost $\opt{\SC}(\rateprof)$. 
The result then follows by noting that the minimal value of the Beckmann potential is continuously differentiable with its gradient equal to the vector equilibrium costs, which are continuous (see \cref{lem:continuity}).
\end{proof}

The following technical lemma is 
used in the proof of \cref{th:differentiability-on-curve}.
In addition to the optimization problem \eqref{eq:Beckmann-fixed-regime-relaxed} we consider the optimal value function for the  perturbed minimization problem
\begin{equation}
\label{eq:Beckmann-fixed-regime-perturbed}
\tag{$\prim_{\rateprof,\pertxprof,\pertfprof}^{\regime}$}
\begin{split}
\valueV_{\regime}(\rateprof,\pertxprof,\pertfprof)=&\min_{(\loadprof,\flowprof)}\potential(\loadprof)
\\
&\text{ s.t. }  
 \sumflow \flowprof=\rateprof,\;\loadprof=\edgepaths \flowprof+\pertxprof, \text{ and } 
\flow_{\route}=\pertf_{\route}\text{ for all } \route\notin\regime,
\end{split}
\end{equation}
where $\pertxprof$ is a perturbation of the load vector and $\pertfprof$ is the flow perturbation vector. 

\begin{lemma}
\label{lem:differentiability-of-fixed-regime-possibly-negative-flows}
Let $(\graph,\costprof,\ODpairs)$ be a routing game structure and  $\regime$ a  given fixed regime. 
Then,
\begin{enumerate}
[label={\textup{(\alph*)}}, ref={\textup{(\alph*)}}]
\item
\label{it:lem:differentiability-of-fixed-regime-possibly-negative-flows-a}
For each $\rateprof\in\Rpos^\ODpairs$ 
the minimum in \eqref{eq:Beckmann-fixed-regime-relaxed} is attained at some $\loadprof\in{\loads}_{\rateprof}^{\regime}$. 
The edge costs $\multedge_{\edge} \coloneqq \cost_{\edge}(\load_{\edge})$ are the same for every optimum  $\loadprof$, and for each $\odpair\in\ODpairs$ there exists a path cost $\multod^{(\odpair)}$ such that $\sum_{\edge\in\route}\multedge_{\edge} = \multod^{(\odpair)}$ for all $\route\in\routes^{(\odpair)}\cap\regime$.  
Moreover, 
the optimal value function $\valueV_{\regime}(\argdot)$ is everywhere finite, convex, and differentiable at $(\rateprof,\zerovec,\zerovec)$, with $\nabla \valueV_\regime(\rateprof,\zerovec,\zerovec)=(\multodprof,\multedgeprof,\multpathprof)$ where
\begin{equation}
\label{eq:mult-path-edge-od}
\multpath_{\route}=\sum_{\edge\in\route}\multedge_{\edge}-\multod^{(\odpair)}\quad
\forall \route\in\routes^{(\odpair)}\setminus\regime.
\end{equation}

\item 
\label{it:lem:differentiability-of-fixed-regime-possibly-negative-flows-b}
If the cost functions are strictly increasing, then \eqref{eq:Beckmann-fixed-regime-relaxed} has a unique optimal solution $\loadprof_{\regime}(\rateprof)$. 
Moreover, if the costs  $\cost_{\edge}$ are $\mathcal C^{1}$ with strictly positive derivative, then $\rateprof\rightarrow\loadprof_{\regime}(\rateprof)$ is also $\mathcal C^{1}$.
\end{enumerate}
\end{lemma}
\begin{proof}
\ref{it:lem:differentiability-of-fixed-regime-possibly-negative-flows-a}
In view of our extension of the costs $\cost_{\edge}(\argdot)$ to $\reals_-$,  we have $\lim_{x\to-\infty}\cost_{\edge}(\load)<0$ and also $\lim_{\load\to\infty}\cost_{\edge}(\load)>0$. Hence, using recession analysis, for each nonzero direction $\dprof\in\reals^\edges\setminus\{0\}$ we have
\begin{equation}
\label{eq:potential}
\potential^{\infty}(\dprof)=\lim_{\var\to\infty}\potential(\var\dprof)/\var=
\sum_{\edge\in\edges} \lim_{\var\to\infty}\cost_{\edge}(\var \dvar_{\edge})\dvar_{\edge}>0,
\end{equation}
so that $\potential$ is inf-compact, and therefore the set of minima of \eqref{eq:Beckmann-fixed-regime-perturbed} is nonempty  (see for example \cite[Theorem 9.2]{Roc:PUP1997}). 
This shows in particular that
the value function $\valueV_{\regime}(\argdot)$ is finite everywhere. 
Hence, by  convex duality, this function is convex and the subdifferential $\partial \valueV_\regime(\rateprof,\zerovec,\zerovec)$ coincides with the optimum solution
set of the dual problem. We next characterize this dual and
prove that it has a unique solution, from where we deduce that it is differentiable with $\nabla \valueV_\regime(\rateprof,\zerovec,\zerovec)$ given as in the statement of part \ref{it:lem:differentiability-of-fixed-regime-possibly-negative-flows-a}.

To write the dual problem we write explicitly the problem \eqref{eq:Beckmann-fixed-regime-relaxed} as
\begin{align}
\label{eq:explicit-relaxed-problem}
\min_{\loadprof,\flowprof}
\sum_{\edge\in\edges}\Cost_{\edge}(\load_{\edge})
\quad\text{s.t.}\
&\sum_{\route\in\routes^{(\odpair)}}\flow_{\route}=\rate^{(\odpair)}
\ \text{for all }\odpair\in\ODpairs,
\ \text{and }\flow_{\route}=0
\ \text{for }\route\notin\regime,\\
&\load_{\edge}=\sum_{\route\in\routes}\edgepath_{\edge\route}\flow_{\route}\ \text{for all }\edge\in\edges\nonumber.
\end{align}

Introducing multipliers $\multodprof=(\multod^{(\odpair)})_{\odpair\in\ODpairs}$, $\multpathprof=(\multpath_{\route})_{\route\not\in\regime}$, and  $\multedgeprof=(\multedge_{\edge})_{\edge\in\edges}$, 
and the  Lagrangian
\begin{equation}
\begin{split}    
\lagrang(\loadprof,\flowprof,\multodprof,\multpathprof,\multedgeprof)
&=
\sum_{\edge\in\edges} \Cost_{\edge}(\load_{\edge})
+\sum_{\odpair\in\ODpairs}\multod^{(\odpair)} \parens*{\rate^{(\odpair)}-\
\sum_{\route\in\routes^{(\odpair)}}\flow_{\route}}
-\sum_{\route\not\in \regime}\multpath_{\route} \flow_{\route}\\
&\quad
+\sum_{\edge\in\edges}\multedge_{\edge} \parens*{\sum_{\route\in\routes}\edgepath_{\edge\route}\flow_{\route}-\load_{\edge}}\\
&=
\sum_{\edge\in\edges} \Cost_{\edge}(\load_{\edge})-\multedge_{\edge} \load_{\edge}+\sum_{\odpair\in\ODpairs}\multod^{(\odpair)}\rate^{(\odpair)}+
\sum_{\route\in\regime} \flow_{\route} \parens*{\sum_{\edge\in \route}\multedge_{\edge}-\multod^{(\odpair(\route))}}\\
&\quad
+\sum_{\route\notin\regime} \flow_{\route}\parens*{\sum_{\edge\in \route}\multedge_{\edge}-\multod^{(\odpair(\route))}-\multpath_{\route}},
\end{split}
\end{equation}
where $\odpair(\route)$ is the commodity of path $\route$, the dual problem becomes
\begin{equation}
\label{eq:dual}
\sup_{\multodprof,\multpathprof,\multedgeprof}\min_{\loadprof,\flowprof}\lagrang(\loadprof,\flowprof,\multodprof,\multpathprof,\multedgeprof).    
\end{equation}
The inner minimum over $(\loadprof,\flowprof)$ can be solved explicitly to obtain the dual in final form as
\begin{equation}
\label{eq:relaxed-dual-final-problem}
\begin{split}
\sup_{\multodprof,\multpathprof,\multedgeprof}
-\sum_{\edge\in\edges}\Cost_{\edge}^{*}(\multedge_{\edge})
+\sum_{\odpair\in\ODpairs}\multod^{(\odpair)}\rate^{(\odpair)}
\quad\text{s.t.}\quad
&
\multod^{(\odpair(\route))}=\sum_{\edge\in \route}\multedge_{\edge}~~\text{for every } \route\in\regime,\\
&\multod^{(\odpair(\route))}=\sum_{\edge\in \route}\multedge_{\edge}-\multpath_{\route}~~\text{for every } \route\notin\regime,     
\end{split}
\end{equation}
where $\Cost_{\edge}^{*}$ is the Fenchel conjugate of $\Cost_{\edge}$.

This amounts to solving 
\begin{equation}
\label{eq:relaxed-dual-problem-simplified}
\min_{\multodprof,\multedgeprof}
\sum_{\edge\in\edges}\Cost_{\edge}^{*}(\multedge_{\edge})
-\sum_{\odpair}\multod^{(\odpair)}\rate^{(\odpair)}
\quad\text{s.t.}\quad
\sum_{\edge\in\route}\multedge_{\edge}=\multod^{(\odpair(\route))} ~~\text{for every } \route\in\regime
,
\end{equation}
and then defining
\begin{equation}
\label{eq:nu-formula-in-relaxed-problem}
\multpath_{\route}=\sum_{\edge\in\route}\multedge_{\edge}-\multod^{(\odpair(\route))}~~\text{for every } \route\not\in\regime.
\end{equation}
As mentioned above, from general convex duality, the optimum solution set of this dual coincides with the sub-differential $\partial \valueV_{\regime}(\rateprof,\zerovec,\zerovec)$ of the primal optimum value function. Since $\valueV_{\regime}(\argdot)$ is finite everywhere, this sub-differential is nonempty and the dual has optimum solutions. 
On the other hand, since $\Cost_{\edge}^{*}$ is strictly convex, the dual optimum solution $(\multodprof,\multpathprof,\multedgeprof)$ is unique  and therefore $\valueV_\regime(\argdot)$ is differentiable at $(\rateprof,\zerovec,\zerovec)$ with $\nabla \valueV_\regime(\rateprof,\zerovec,\zerovec)=(\multodprof,\multpathprof,\multedgeprof)$.
Writing the optimality conditions for \eqref{eq:explicit-relaxed-problem} we obtain $\cost_{\edge}(\load_{\edge})=\multedge_{\edge}$, and the characterization of the gradient follows from \eqref{eq:relaxed-dual-problem-simplified} and \eqref{eq:nu-formula-in-relaxed-problem}.
This establishes part \ref{it:lem:differentiability-of-fixed-regime-possibly-negative-flows-a} of the lemma.

\vspace{1ex}
\ref{it:lem:differentiability-of-fixed-regime-possibly-negative-flows-b} 
Note that the relaxed constraints in \eqref{eq:Beckmann-fixed-regime-relaxed} are just linear equality constraints. 
Projecting these equations onto the space of pairs $(\loadprof,\rateprof)$, one can remove the flow variables $\flowprof$ and rewrite the constraints as a
linear system of equations in the load variables $\loadprof$ only, that is,
\begin{equation}
\label{eq:linear-A-B} 
\Amatr\loadprof+\Bmatr\rateprof= \zerovec,
\end{equation}
where the matrices $\Amatr$ and $\Bmatr$ depend on the regime $\regime$, and the matrix $\begin{bmatrix}
\Amatr & \Bmatr    
\end{bmatrix}$ has linearly independent rows.
Hence, the problem can be restated in the equivalent form
\begin{equation}
\min_{\loadprof\in\mathbb R^{\edges}}\{\potential(\loadprof) \colon \Amatr\loadprof+\Bmatr\rateprof=\zerovec\}.
\tag{$\prim_{\rateprof}^{\regime}$}
\end{equation}

Introducing a Lagrange multiplier vector $\multAprof$, the optimality conditions are
\begin{equation}
\label{eq:optimality-conditions-withAB}
\begin{split}
\nabla\potential(\loadprof) + \Amatr^{\top}\multAprof 
&= \zerovec,\\
\Amatr\loadprof+\Bmatr\rateprof &= \zerovec.
\end{split}
\end{equation}

We note that for each $\rateprof \in\Rpos^\ODpairs$ the (unique) optimum solution $\loadprof=\loadprof(\rateprof)$  also has a unique corresponding multiplier $\multAprof = \multAprof(\rateprof)$. 
This follows from $\Ker \Amatr^{\top}=\{\zerovec\}$.
To prove the latter, suppose that $\Amatr^{\top}\multBprof = \zerovec$ for some vector $\multBprof$. 
For each $\rateprof\in\Rpos^{\ODpairs}$ we can find a 
feasible $\loadprof\in\loads_{\rateprof}$ such that $\Amatr\loadprof+\Bmatr\rateprof = \zerovec$, by just routing any flow of demand $\rateprof$ that uses the regime $\regime$. 
We have then
\begin{equation}
\label{eq:A-B}    
0 = \inner*{\multBprof}{\Amatr\loadprof+\Bmatr\rateprof} 
= \inner*{\Bmatr^{\top}\multBprof}{\rateprof}.
\end{equation}
for all $\rateprof\in\Rpos^{\ODpairs}$, from which we deduce $\Bmatr^{\top}\multBprof=\zerovec$. 
Hence
$\begin{bmatrix}
\Amatr & \Bmatr    
\end{bmatrix}^{\top} \multBprof = \zerovec$ and therefore $\multBprof = \zerovec$, because the rows of $\begin{bmatrix}
\Amatr & \Bmatr    
\end{bmatrix}$ 
are linearly independent.

The conclusion of the lemma follows from the implicit function theorem if we show that the left hand side of \eqref{eq:optimality-conditions-withAB} has a nonsingular Jacobian at any point $(\loadprof, \multAprof)$.

This amounts to proving that the following homogeneous linear system 
\begin{equation}
\label{eq:homogenous-jacobian-system}
\begin{split}
\nabla^2\potential(\loadprof) \dprof + \Amatr^{\top}\multBprof &= \zerovec\\
\Amatr \dprof &= \zerovec
\end{split}
\end{equation} 
admits only the trivial solution.
Multiplying the first equation by $\dprof$ and using the second equation it follows that 
\begin{equation}
\label{eq:singular-jacobian-condition-2}
0 = \inner*{\dprof}{\nabla^2\potential(\loadprof) \dprof}
+ \inner*{\dprof}{\Amatr^{\top}\multBprof}
= \inner*{\dprof}{\nabla^2\potential(\loadprof) \dprof}
=\sum_{\edge\in\edges}\cost_{\edge}'(\load_{\edge})\dvar_{\edge}^2,
\end{equation}
which implies $\dprof=\zerovec$ because $\cost_{\edge}'(\load_{\edge})>0$. 
Hence the first equation reduces to $\Amatr^{\top}\multBprof=\zerovec$, which implies $\multBprof=\zerovec$, because $\Ker(\Amatr^{\top})=\{\zerovec\}$, as shown above. 
This completes the proof.
\end{proof}
    
\begin{remark}
\label{rm:differentiability-with-cycles-condition}
\cref{lem:differentiability-of-fixed-regime-possibly-negative-flows,th:differentiability-on-curve} can be generalized to games with  cost functions that are just strictly increasing, assuming that the set of edges $\edge$ where $\cost'(\load_{\edge}(\rateprof))=0$ form a graph without undirected cycles. 
Indeed, \eqref{eq:homogenous-jacobian-system} implies that for every active edge $\edge$, either $\dvar_{\edge}=0$ or $\cost_{\edge}'(\load_{\edge})=0$, whereas the condition $\Amatr\dprof=\zerovec$ tells us that the vector $\dprof$ is a (possibly negative) flow of zero demand on our network. 
Since $\dvar_{\edge}=0$ for all edges such that $\cost_{\edge}'(\load_{\edge})>0$,  the vector $\dprof$ induces a flow of zero demand on the subnetwork formed by edges $\edge$ such that $\cost_{\edge}'(\load_{\edge})=0$. If the latter subnetwork has no undirected cycles, then we can conclude that $\dprof=\zerovec$ and complete the proof as above.
\end{remark}

\section{List of symbols}
\label{se:symbols}

\begin{longtable}{p{.15\textwidth} p{.83\textwidth}}

$\cost_{\edge}$ & cost of edge $\edge$\\
$\costprof$ & edge cost function vector\\
$\cost_{\route}$ & cost of path $\route$\\
$\Cost_{\edge}(\load_{\edge})$ & $\int_0^{\load_{\edge}}\cost_{\edge}(z)\diff z$\\
$\Cost_{\edge}^{*}(\argdot)$ & Fenchel conjugate of $\Cost_{\edge}(\argdot)$\\
$\sink^{(\odpair)}$ & destination of \ac{OD} pair $\odpair$\\
$\edge$ & edge\\
$\edges$ & set of edges\\
$\flow_{\route}$ & flow on path $\route$\\
$\eq{\flow}_{\route}$ & equilibrium flow on path $\route$\\
$\flowprof$ & flow vector\\
$\eq{\flowprof}$ & equilibrium flow vector\\
$\flows_{\rateprof}$ & $ \braces*{\flowprof\in\Rpos^{\routes}\colon 
\sum_{\route\in\routes^{(\odpair)}}\flow_{\route}=\rate^{(\odpair)}\text{ for all }\odpair\in\ODpairs}$, set of feasible flows, defined in \eqref{eq:flows}\\
$\graph$ & $ (\vertices,\edges)$, directed multigraph with set of vertices $\vertices$ and set of edges $\edges$ \\

$\ODpairs$ & set of \ac{OD} pairs \\

$\nOD$ & number of elements in $\ODpairs$\\
$\inter^{-}$ & $ (\bar\var-\varepsilon,\bar\var)$, defined in \cref{le:GeneralWatlingLemma} \\ 
$\inter^{+} $ & $ (\bar\var, \bar\var+\varepsilon)$, defined in \cref{le:GeneralWatlingLemma}\\
$\jac_{\costprof}(\argdot)$ & Jacobian matrix of the path costs function $\costprof:\Rpos^n\rightarrow\Rpos^n$\\
$\multod^{(\odpair)}$ & Lagrange multiplier associated to \ac{OD} pair $\odpair$\\
$\multodprof$ & \ac{OD}-pairs Lagrange multiplier vector\\
$\neigh$ & neighborhood\\
$\source^{(\odpair)}$ & origin of \ac{OD} pair $\odpair$\\
$\route$ & path\\
$\routes^{(\odpair)}$ &   set of paths of \ac{OD} pair $\odpair$\\

$\npaths^{(\odpair)}$ &  number of elements in $\routes^{(\odpair)}$\\
$\npaths$ & number of elements in $\routes$ \\
$\routes$ & 
$\union_{\odpair\in\ODpairs} \routes^{(\odpair)}$,  union of the path sets of every \ac{OD} pair, defined in \eqref{eq:paths}\\
$\activ{\routes}(\rateprof)$ & active regime, defined in \eqref{eq:active-regime}   \\
$\activ{\routes}(\bar\var^{-})$ & $ \activ{\routes}(\rateprof(\bar\var-\varepsilon))$, defined in \eqref{eq:left-right-regime} \\
$\activ{\routes}(\bar\var^{+})$ & $  \activ{\routes}(\rateprof(\bar\var+\varepsilon))$, defined in \eqref{eq:left-right-regime} \\
$\PoA(\rateprof)$ & \acl{PoA} with demand $\rateprof$, defined in \eqref{eq:poa}\\
$\odrate^{(\odpair)}$ & rate of \ac{OD} pair $\odpair$ for linearly increasing demand\\
$\odrateprof$ & rate vector for linearly increasing demand\\
$\regime$ & regime, defined in \cref{de:regime}\\ 
$\regime^{-}$ & $ \activ{\routes}(\bar\var^{-})$, defined in \cref{le:GeneralWatlingLemma} \\
$\regime^{+}$ & $  \activ{\routes}(\bar\var^{+})$, defined in \cref{le:GeneralWatlingLemma} \\
$\sumflow$  & defined in \eqref{eq:matricesConstraints}\\
$\SC$ & social cost, defined in \eqref{eq:social-cost}\\
$\eq{\SC}(\rateprof)$ & equilibrium social cost  with demand $\rateprof$, defined in \eqref{eq:equilibrium-social-cost}\\
    
$\opt{\SC}(\rateprof)$ & optimum social cost with demand $\rateprof$, defined in \eqref{eq:optimum-social-cost}\\
$\vertices$ & set of vertices\\
$\load_{\edge}$ & load on edge $\edge$\\
$\loadprof$ & load vector\\
$\loadprof_{\regime}(\rateprof)$ & optimum solution of \eqref{eq:Beckmann-fixed-regime-relaxed}\\
$\loads_{\rateprof}$ & set of load profiles induced by some flow of demand $\rateprof$\\
$\yvar_{\route}$ & introduced in \eqref{eq:program_fprime_+}\\
$\yvec$ & introduced in \eqref{eq:quad_opt}\\

$\multAprof$ & Lagrange multiplier vector\\
$\multBprof$ & Lagrange multiplier vector\\
$\canbasis_{\route}$ & $\route$-th element of the canonical basis\\ 
$\edgepath_{\edge\route}$
& $
\begin{cases}
1 & \text{if }\edge\in \route,\\
0 & \text{otherwise}.    
\end{cases}$, defined in \eqref{eq:loads}\\

$\edgepaths$ & defined in \eqref{eq:matricesConstraints}\\
$\multedge_{\edge}$ & Lagrange multiplier associated to edge $\edge$\\
$\multedgeprof$ & edge Lagrange multiplier vector\\
$\thfunc(\var)$ & 
$\sum_{\odpair\in\ODpairs}
\odrate^{(\odpair)} \cdot(\eqcost^{(\odpair)})'(\var)$, defined in \cref{eq:des}\\

$\eqcost^{(\odpair)}(\rateprof)$ & equilibrium cost of \ac{OD} pair $\odpair$ with demand $\rateprof$ \\
$\eqcostprof$ & equilibrium cost vector\\
$\rate^{(\odpair)}$ & demand at \ac{OD} pair $\odpair$\\
$\rateprof$ & $\parens*{\rate^{(1)},\dots,\rate^{(\nOD)}}$, demand vector\\
$\multpath_{\route}$ & Lagrange multiplier associated to path $\route$\\
$\multpathprof$ & path Lagrange multiplier vector\\
$\pertx_{\edge}$ & perturbation of load  $\load_{\edge}$\\
$\pertxprof$ & perturbation of load vector $\loadprof$\\
$\eqcostedge_{\edge}$ & $\cost_{\edge}(\eq{\load}_{\edge})$ equilibrium cost of edge $\edge$\\

$\potential(\loadprof)$ &$\sum_{\edge\in\edges}\Cost_{\edge}(\load_{\edge})$ potential \\  
$\pertf_{\route}$ & perturbation of flow $\flow_{\route}$\\
$\pertfprof$ & perturbation of flow vector $\flowprof$\\
$\zerovec$ & zero vector\\
$\onevec^{(\odpair)}$ & $ \sum_{\route\in\odpair} \canbasis_{\route}$ \\

\end{longtable}

% Bibliography

\bibliographystyle{apalike}
\bibliography{biblio-games}

\end{document}